%% file: main.tex
\documentclass[aps,
pra, 
showpacs,
notfootinbib,
reprint,
10pt,
longbibliography,
superscriptaddress]{revtex4-2}
\pdfoutput=1

\usepackage[T1]{fontenc}    
\usepackage{lmodern} 
\usepackage{microtype}
\usepackage[english]{babel} 

\usepackage[table,dvipsnames]{xcolor}
\usepackage[linesnumbered,ruled,vlined]{algorithm2e}
\usepackage{hyperref}
\PassOptionsToPackage{linktocpage}{hyperref} 

\usepackage{asymptote}

\usepackage{amsmath,amsfonts,amssymb,amsthm,bbm,graphicx,enumerate,bm,times}
\usepackage{mathtools}
\usepackage{comment}
\usepackage{paralist, tabularx}

\hypersetup{
  hypertexnames=false,
  colorlinks   = true, 
  urlcolor     = magenta, 
  linkcolor    = blue, 
  citecolor   = green 
}
\usepackage{bookmark} 

\usepackage[
    capitalize
    ]{cleveref} 
\usepackage[printonlyused,withpage,nolist]{acronym} 

\usepackage{tikz}
\usepackage{graphicx}
\usepackage{bm}
\usetikzlibrary{shapes,positioning}
\usepackage{etoolbox}
\usepackage{todonotes}
\usepackage{setspace}

 \newtheorem{theorem}{Theorem}
 \newtheorem{proposition}[theorem]{Proposition}

 \newtheorem{lemma}[theorem]{Lemma}
 
 \newtheorem{definition}[theorem]{Definition}

\AddToHook{env/lemma/begin}{\crefalias{theorem}{lemma}}
\AddToHook{env/proposition/begin}{\crefalias{theorem}{proposition}}
\AddToHook{env/definition/begin}{\crefalias{theorem}{definition}}

\newcommand{\e}{\mathrm{e}}

\newcommand{\Tr}{\mathrm{Tr}} 

\newcommand{\poly}{\mathrm{poly}}

\newcommand{\vertiii}[1]{{\left\vert\kern-0.25ex\left\vert\kern-0.25ex\left\vert #1 
    \right\vert\kern-0.25ex\right\vert\kern-0.25ex\right\vert}}

\newcommand{\exact}{\text{\rm exact}}

\newcommand{\ZZ}{{\mathbb Z}}

\DeclarePairedDelimiterX\norm[1]\lVert\rVert{%
  \ifblank{#1}{\,\cdot\,}{#1}
}   

\DeclarePairedDelimiterX\Set[1]\{\}{%
  
  #1
}

\DeclarePairedDelimiterX\braket[2]{\langle}{\rangle}%
  {#1\kern0.15ex\delimsize\vert\kern0.15ex\mathopen{}#2}
\DeclarePairedDelimiterX\ketbra[2]{\vert}{\vert}%
  {#1\kern0.15ex\delimsize\rangle\delimsize\langle\kern0.15ex\mathopen{}#2}
\DeclarePairedDelimiterX\sandwich[3]{\langle}{\rangle}%
  {#1\,\delimsize\vert\kern0.15ex\mathopen{}#2\kern0.15ex\delimsize\vert\kern0.15ex\mathopen{}#3}

\renewcommand{\vec}[1]{\mathbf{#1}}

\newcommand{\xfeas}{\ensuremath\vec x_\mathrm{\rm feas}}%
\newcommand{\ful}{\ensuremath f_{\widehat{\mathrm{\rm unc}}}}%

\definecolor{ingo}{rgb}{.8,.5,0}
\definecolor{edoardo}{rgb}{0.5,0,0.5}
\definecolor{sergi}{rgb}{1,0,.5}
\definecolor{martin}{rgb}{0,.4,1}
\definecolor{stefan}{rgb}{0.7,.7,1}

\definecolor{aifeedback}{rgb}{0,.55,.3}

\newcommand{\tii}{Quantum Research Center, Technology Innovation Institute (TII), Abu Dhabi}

\newcommand{\tuhh}{Hamburg University of Technology, Institute for Quantum Inspired and Quantum Optimization, Hamburg, Germany}

\newcommand{\fujitsu}{Fujitsu Germany GmbH, Mies-van-der-Rohe-Straße 8, 80807 Munich, Germany}

\begin{document}

\title{
Scalable Determination of Penalization Weights\texorpdfstring{\\}{ }for Constrained Optimizations on Approximate Solvers}

\begin{abstract}
    Quadratic unconstrained binary optimization (QUBO) provides problem formulations for various computational problems that can be solved with dedicated QUBO solvers,  based on classical or quantum computation. 
    A common approach to constrained combinatorial optimization problems is to enforce the constraints in the QUBO formulation by adding penalization terms. Penalization introduces an additional hyperparameter that significantly affects the solver's efficacy: the relative weight between the objective terms and the penalization terms. 
    Here we study the problem of a priori determining suitable penalization weights from a foundational and practical perspective. 
    We formally establish that finding its optimal value for Gibbs solvers is $\mathsf{NP}$-hard even approximately and $\#\mathsf{P}$-hard exactly.
    This serves us as a motivation to develop a widely applicable heuristic strategy for determining suitable penalization weights, 
    with provable guarantees for Gibbs solvers, and polynomial complexity for broad problem classes. 
    Experiments across diverse problems and solver architectures, including large-scale instances on Fujitsu’s Digital Annealer, show robust performance and order-of-magnitude speedups over existing heuristics.
    We further extend the strategy to quantum approximate optimization and demonstrate its performance on Quantinuum's H2-2 trapped-ion processor.
\end{abstract}

\author{Edoardo Alessandroni$^*$}
\affiliation{\tii}
\affiliation{SISSA — Scuola Internazionale Superiore di Studi Avanzati, Trieste, Italy}
\email{Corresponding author email: ealessan@sissa.it}
\author{Sergi Ramos-Calderer}
\affiliation{Centre for Quantum Technologies, National University of Singapore, Singapore.}
\affiliation{\tii}
\author{Michel Krispin}
\affiliation{\tuhh}
\author{Markus Kirsch}
\affiliation{\fujitsu}
\author{Christian Münch}
\affiliation{\fujitsu}
\author{Fritz Schinkel}
\affiliation{\fujitsu}
\author{Stefan Walter}
\affiliation{\fujitsu}
\author{Martin Kliesch}
\affiliation{\tuhh}
\author{Leandro Aolita}
\affiliation{\tii}
\author{Ingo Roth}
\affiliation{\tii}

\maketitle

Combinatorial optimization problems arise widely in both theoretical research and real-world applications, aiming to identify the optimal configuration of discrete decision variables that minimizes a given objective function.
The corresponding search space typically has a size scaling exponentially with the number of variables, rendering the optimization a computationally hard problem.
Methods for solving these problems have been a subject of intense study in both classical and quantum computation \cite{Abbas24ChallengesAndOpportunities}. 
In particular, powerful heuristics and dedicated hardware have been developed \cite{Johnson11QuantumAnnealingWith, Crosson21ProspectsForQuantum, 
FarGolGut14, 
Inagaki16CoherentIsingMachine, Honjo21_100k-spin, 
Aao19ApplicationOfDigital, Hideki23MathematicalAspectsOf, 
Okuyama19BinaryOptimizationBy, 
Goto19CombinatorialOptimization, Aatsumura21ScalingOutIsing} to tackle \ac{QUBO} problems \cite{kochenberger14TheUnconstrainedBinary, blog_qubo_list_2021}.
These include approaches based on Gibbs sampling or simulated annealing \cite{Kirkpatrick_SA_83, Geman_SA_84, Glauber_MCMC_63, Hastings_MCMC_70, Hukushima_EMC_96, Mossel_Gibbs_13, Siddique_SA_16, karabin_SA_20}, special-purpose chips for digital annealing \cite{da-1, da-2, da-3}, and quantum primitives such as adiabatic evolution \cite{Kadowaki_QA_98, Farhi_QA_01, Rajak_QA_22, Munoz-Bauz_QA_25} and different variants of the \ac{QAOA} \cite{Farhi_QAOA_14, Blekos_QAOA_23, Cheng_QAOA_24, Amosy_QAOA_24, Yanakiev_QAOA_24}.

Practical solvers tend to be heuristic algorithms that yield approximate solutions. 
For instance, simulated annealing (or any solver based on Gibbs sampling) outputs a configuration sampled from a low-temperature thermal distribution over the combinatorial space,  
with an objective value close (but in general not equal) to the optimum. 
The colder the distribution, the higher the probability that the output state is very close to the optimal solution of the problem.
More precisely, for ideal Gibbs samplers, each configuration $\vec x$ is sampled with a probability proportional to $\e^{-E(\vec x)\beta}$, 
where $E(\vec x)$ is the energy of the configuration and $\beta$ is the inverse temperature of the system.
Sampling from this probability distribution or estimating the associated partition function is generally hard \cite{Jerrum_Ising_93, Sly_twospin_14}, so practical solvers rely on Markov Chain Monte Carlo techniques \cite{Gilks_MCMC_96, Kirkpatrick_SA_83} and annealing schedules to approximate it. These methods, such as simulated annealing, propose new configurations and accept or reject them according to their energy and target temperature, gradually biasing the search toward low-energy configurations. Hardware-accelerated implementations, such as Fujitsu's Digital Annealer~\cite{da-1, da-2, da-3}, further enhance the exploration of the energy landscape through fast, parallelized sampling mechanisms, thereby improving the chances of finding solutions with low objective.

\def\Binf{\ensuremath B_{
\!\bar{\mathcal F}}}
\def\Bfl{\ensuremath  B^{\scalebox{.7}{<}}_{\!\mathcal F}}
\def\Bfh{\ensuremath  B^{\scalebox{.7}{>}}_{\!\mathcal F}}
\def\epsl{\ensuremath \epsilon^{\scalebox{.7}{<}}}
\def\epsh{\ensuremath \epsilon^{\scalebox{.7}{>}}}

\begin{figure*}[htbp]
    \centering
        \includegraphics[width=\linewidth]{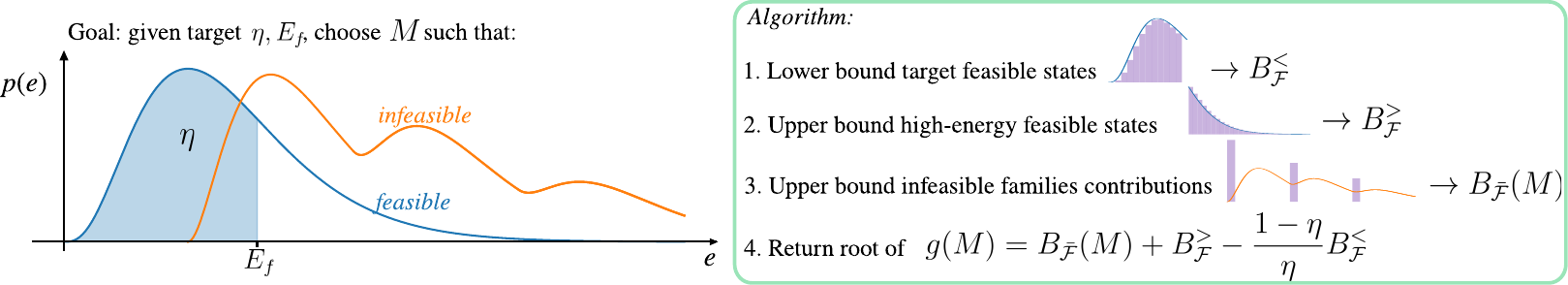}
    \caption{\emph{The Big-$M$ problem for approximate solvers}
    is to ensure by the choice of a penalization weight $M$ that an approximate solver samples feasible solutions with probability at least $\eta$ when given a QUBO reformulation of a constrained optimization problem. 
    Optionally, one can additionally enforce the solutions to be below a certain energy threshold $E_f$. 
    We assume that the output distribution of a solver is qualitatively approximated by a Gibbs distribution at known inverse temperature, illustrated to the left in terms of the probability $p(e)$ of sampling a solution with energy $e$ conditioned on the solution being feasible or infeasible.  
    The density of infeasible solutions is naturally grouped into families (`humps'), each one characterized by $E^{(p)}(x)$ taking a certain value.
    Our method to determine $M$, summarized to the right, calculates (1.) a lower bound $\Bfl$ and (2.) an upper bound $\Bfh$ on the probabilities of sampling feasible points with objective below and exceeding $E_f$, respectively. 
    Together with (3.) an upper bound $\Binf$ on the probability of infeasible events, (4.) the penalization weight $M$ is determined as the unique root of the scalar function $g(M)$, that depends on the targeted success probability $\eta$. 
    We argue that this method is efficient for large classes of problems, prove theoretical guarantees on its performance, and demonstrate its practical applicability numerically.
    }
    \label{fig:fig1}
\end{figure*}

Importantly, general combinatorial optimization problems, arising in applications, usually include constraints that restrict the search space. 
Yet, a standard approach to tackle such problems with \ac{QUBO} solvers is to convert constraints into penalization terms weighted by a constant, commonly referred to as Big-$M$. 
The choice of this constant crucially shapes the energy landscape.
If the constant $M$ is set too high, the low-energy spectrum becomes uninformative: the solver is forced to prioritize constraint satisfaction at all costs, returning feasible states that may be far from optimal with respect to the original objective function. Conversely, if $M$ is chosen too small, the energies of infeasible configurations may lie near, or even below, the optimal feasible region, causing approximate solvers to sample disproportionately from infeasible solutions that violate constraints. 
Existing computationally-efficient Big-$M$ prescriptions tend to substantially overestimate the required penalty \cite{Harwood_routing, Leonidas_vehiclerouting, qiskit_doc, Garcia_penalty_22}, which in practice degrades solution quality \cite{Harwood_routing, Azad_22_VHP, Roch_CE_21, BigM_25}. Although recent work \cite{BigM_25} proposes a practical strategy that delivers significantly lower (but still sufficient) Big-$M$ values, it is primarily designed for exact solvers. 
Current approaches do not incorporate the degree of approximation characteristic of modern heuristic methods, such as a Gibbs sampler at finite temperature.
A systematic penalization strategy for approximate solvers is missing.

In this work, we introduce a novel, broadly applicable algorithm for a priori determining the penalization term for a given constrained optimization problem and a specified approximate solver. 
Our approach combines analytical considerations with uniform sampling over feasible configurations, to derive efficiently evaluable bounds on the solver's output distribution, from which the penalization weight $M$ is calculated.
We prove that for exact Gibbs solvers at arbitrary $\beta$, the algorithm yields, with high probability, a \ac{QUBO} reformulation for which the probability of sampling feasible solutions with energy at most $E_f$ is guaranteed be greater than a controllable input parameter.
The algorithm's hyperparameters allow trading off run-time against the deviation of the returned upper bound from the optimal (minimal) $M$. 

We further show that, for large classes of constrained optimization problems and appropriate hyperparameters, the algorithm's runtime and memory scales polynomially in the system size.
We numerically demonstrate the practical applicability of our method across relevant parameter regimes,  diverse problem instances and solver architectures. 
In particular, we benchmark the approach on representative constrained optimization problems, including the \ac{TSP}, the \ac{MNPP}, and \ac{PO}.
Besides small-scale evaluations for exact Gibbs sampling and intermediate-scale experiments with simulated annealing, 
we show that our method can be used to determine penalization weights for Fujitsu's Digital Annealer (version 3) on problem instances of up to several thousand bits. 
Although the Digital Annealer is known to deviate from our underlying assumption of thermal output distributions \cite{Hideki23MathematicalAspectsOf}, we find that our method still qualitatively captures its behavior sufficiently well to achieve an order-of-magnitude speed-up in time-to-solution compared to direct binary searches for $M$ based on simpler heuristics.

Finally, we apply our method to quantum solvers, where the high cost of individual solver calls makes classical pre-computation of the penalization weight particularly valuable. We consider the \ac{QAOA} and model the output distribution of a depth-$p$ circuit as thermal, at an effective inverse temperature taken from recent empirical characterizations of optimized \ac{QAOA} ensembles \cite{DiezValle_Teff_25}.
In small scale experiments, the resulting weights meet the requested feasibility targets in noiseless state-vector simulations, and remain effective on Quantinuum's H2-2 trapped-ion quantum processor \cite{h2_system} in presence of the processor's noise. The sampled states concentrate significantly more at feasible solutions with low objective values than a standard heuristic.

\begin{figure*}[htbp]
    \centering
        \includegraphics[width=\linewidth]{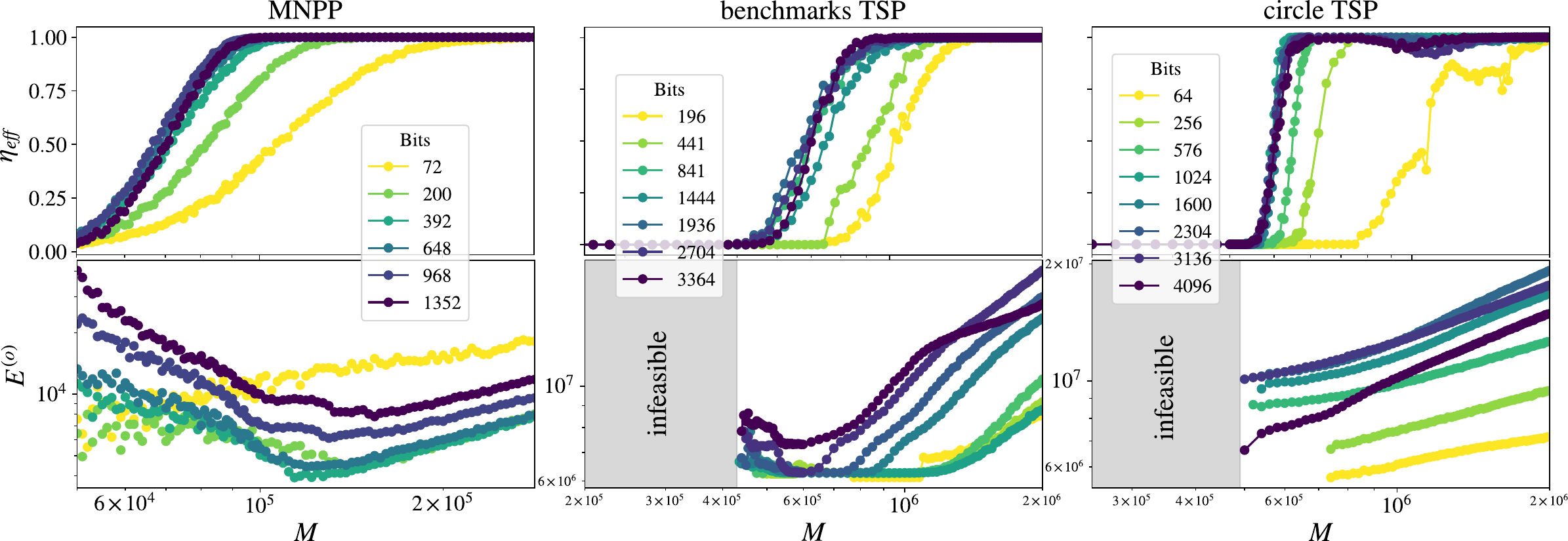}
    \caption{\label{fig:bigM_prob} Proportion of feasible solutions observed $\eta_\mathrm{eff}$(top) and mean objective energy $E^{(o)}$ of sampled feasible solutions (bottom) on the \ac{DA} solver (version 3) for different benchmarked problems (from left to right: \acf{MNPP}, \acf{TSP} instances from library \cite{TSPbenchmarks}, and \ac{TSP} with cities placed on a circle) with different values of $M$ and problem size. The gray areas lack mean energy points because only infeasible bitstrings were sampled for those values of $M$. In the top panels, we observe a phase transition from infeasible to feasible solutions as $M$ increases, indicating that selecting optimized values of $M$ is needed. 
    However, on the bottom panels, we notice a degradation in the quality of the sampled bitstrings.
    The mean of the energy $E^{(o)}$ of the outputs of the sampler increases for larger values of $M$, beyond a seemingly sweet spot that is located around the transition.
    For reference, the $M$ values suggested by the trivial choice in \cref{eq:M_l1} are several orders of magnitude larger than the shown scales: around $10^9$ for \ac{MNPP}, $10^8$ for benchmarks \ac{TSP} and $10^{10}$ for circle \ac{TSP}. Such extreme overshooting implies that the mean energy sampled by the solver would lie far from the desired minimum, undermining the optimization's effectiveness.
    }
\end{figure*}

\section*{Results}

We start the presentation of our results with formalizing the problem of determining penalization weights and empirically demonstrating the importance of the Big-$M$ problem for approximate solvers using Fujitsu's Digital Annealer (version 3). 
In the subsequent section we describe our algorithm for determining $M$ and establish its theoretical guarantees, before turning to numerical validation and benchmark in the last section.

\subsection{The Big-\texorpdfstring{$M$}{M} problem for approximate combinatorial solvers}

In the following, we consider the constrained optimization problems
\begin{equation}\tag{P}\label{lcBQP}
    \operatorname*{minimize}_{\vec x \in \{0,1\}^n}\ E^{(o)}(\vec x) = \vec x^t Q \vec x 
    \quad\operatorname{subject\ to}\quad A\, \vec x = \vec b
\end{equation}
given by $Q \in \mathbb R^{n\times n}$, 
$A \in \ZZ^{m\times n}$, and $\vec b \in \ZZ^m$.
This special type of \ac{LCBO} problems is able to capture complex problem formulations such as polynomially constrained problems and integer-variable problems, which can all be cast into this form using certain \emph{gadgets} \cite{Williams_EqConstr_13, Glover_Linearization_65, Rosenberg_Quadratization_75, Nemhauser_BinaryMap_88, BigM_25}.

A constrained problem in the form \eqref{lcBQP} can be converted to a \ac{QUBO} problem by promoting the constraints $A\, \vec x = \vec b$ to quadratic penalty terms 
\begin{equation}\label{eq:penalty_energy}
    E^{(p)}(\vec x) = (A\vec x - \vec b)^2,
\end{equation}
weighted by a penalization constant $M>0$. 
The new function to minimize, now a sum of the objective and penalization contributions, reads
\begin{equation}\tag{P$_{\!M}$}\label{QUBO}
    \operatorname*{minimize}_{\vec x \in \{0,1\}^n}\ E(\vec x) = \vec x^t Q \vec x + M(A\vec x - \vec b)^2\,.
\end{equation}
The minimization is now over the entire space $\{0,1\}^n$, but infeasible bit strings will incur an energy penalization.
We consider \eqref{QUBO} an \emph{exact reformulation} of \eqref{lcBQP} when the optimal points remain unchanged. 
Thus, solving an exact reformulation \eqref{QUBO} with an exact combinatorial solver yields an optimal solution to the original constrained problem. 
For an exact solver, its runtime and required computational precision depend on $M$. 
This insight motivates choosing a minimal penalization constant $M_\exact^\ast$ that still ensures an exact reformulation. 
As shown in \cite{BigM_25}, while finding $M_\text{\rm exact}^\ast$ is NP-hard, 
good approximations to it can be found using the following strategy: 
given a feasible point $\xfeas$, a lower-bound $\ful$ on the objective, with $\ful\leq E^{(o)}(x)$ for any $x\in \{0,1\}^n$, and any constant $\delta > 0$, \cref{QUBO} with
\begin{equation}\label{eq:MviaBounds}
    M \coloneqq E^{(o)}(\xfeas) - \ful + \delta\,
\end{equation}
is an exact reformulation of \eqref{lcBQP}.
A feasible point and an objective lower bound can be efficiently pre-computed using, e.g. greedy algorithms and SDP relaxations, respectively. 
This approach leads to a tighter upper bound to $M_\text{\rm exact}^\ast$ than the trivial upper bound $M_\text{\rm exact}^\ast \leq \|Q\|_{\ell_1} + \delta$ \cite{BigM_25}, which in turn improves the run-time of solvers.

However, approximate solvers do not necessarily return the optimal point but a solution with a low objective value approximating the true minimum. 
For this reason, using the strategy described in Ref.~\cite{BigM_25} will generally not ensure feasibility of the solution for an approximate solver.  
At the same time, also for approximate solvers, one expects that choosing a large value for $M$ will rapidly deteriorate the quality of the solver's results, manifesting the so-called \emph{Big-$M$ problem}. 

As a first result, we establish the need for `fine-tuning' $M$ in the case of  Fujitsu's \ac{DAU} for some examples. 
\cref{fig:bigM_prob} shows the frequency of feasible solutions and their mean energy using different values of $M$ for different problem instances and sizes. 
We observe that, in the regime where the majority of solutions are feasible, the mean objective value increases significantly with larger $M$.  
Thus, to ensure low objectives, it is important to choose a value of $M$ close to the transition in the feasibility probability.  
And the `good' regime for $M$ becomes narrower as the system size increases. 
Notice also that a naive choice of $M$, like the one in \cref{eq:M_l1}, overestimates the transition point by several orders of magnitude, and consequently will return solutions with an undesired, high objective value.

These observations motivate us to devise a systematic strategy for solving the \emph{Big-$M$ problem}.  
We begin by formalizing the problem statement with the following definition. 
Let $\mathcal F = \{\vec x \mid A\vec x = \vec b\} \subset \{0,1\}^n$ denote the subspace of feasible points.

\begin{definition}
\label{def:eta-reformulation}
The \ac{QUBO} problem \eqref{QUBO} is an \emph{$\eta$-reformulation of the problem \eqref{lcBQP} for a solver with guaranteed energy threshold $E_f$}, or short \emph{$\eta$-reformulation}, if the solver's output distribution on \eqref{QUBO} fulfills 
\begin{equation} \label{eq:eta_reformulation}
    \Pr[\{x \in \mathcal F \mid E(x) \leq E_f\}] \geq \eta\,.
\end{equation}  
\end{definition}

In other words, such a reformulation ensures observing feasible solutions of a low energy with (at least) constant probability. 
We consider $M$ to be optimal if it is the minimal value guaranteeing an $\eta$-reformulation and denote it as $M^\ast_\eta$. 
As our first result, we establish the principle difficulty of finding 
optimal penalization weights in general.
\begin{theorem}
  Determining $M^\ast_\eta$ up to constant additive error is $\mathsf{NP}$-hard and achieving exponentially small additive error is $\#\mathsf{P}$-hard.
\end{theorem}
We provide the formal problem definition, an extended theorem formulation and the proof reducing the problems to finding the maximum cardinality and counting independent sets of a graph in Appendix~\ref{app:hardness_proof}. 
In face of the hardness, our goal in the remainder is to devise an efficient strategy that at least approximates $M^\ast_\eta$ from above and benchmark the quality of its solution for different instances.

\subsection{A Big-\texorpdfstring{$M$}{M} strategy} \label{sec:recipe}

\def\einf{\overline{\mathcal F}}
\def\efh{\mathcal F \cap \{ E^{(o)} > E_f\}}
\def\efl{\mathcal F \cap \{ E^{(o)} \leq E_f\}}
\def\nv{n_\mathrm{pen}}
\def\ELB{E_\mathrm{LB}}
\def\vcut{v_\mathrm{cut}}

\def\LambdaLow{\ensuremath \Lambda^{\!\scalebox{.7}{<}}}
\def\LambdaHigh{\ensuremath \overline{\LambdaLow}}


An apparent challenge in defining a strategy for solving the Big-$M$ problem is that the definition of an $\eta$-reformulation depends on the actual output distribution of the solver under consideration. 
This is different from the exact reformulation, which only depends on the problem instance itself. 
Output distributions of approximate solvers are generally not known a priori and may depend in complex ways on optimization schedules and other hyperparameters.  
To overcome this obstacle, we consider Gibbs samplers as the prototypical proxies of an approximate solver.
Indeed, a large class of approximate optimization algorithms, including Metropolis-like dynamics, simulated annealing, and more general MCMC-based solvers, are fundamentally grounded in Gibbs sampling principles, as they are designed to progressively concentrate probability mass onto low-energy configurations of an associated Gibbs distribution \cite{Geman_SA_84, Kirkpatrick_SA_83, Metropolis_53}.

A Gibbs sampler at inverse temperature $\beta$ has output distribution $p(\vec x) = \mathcal N_\beta e^{-\beta E(\vec x)}$, where $\mathcal{N}_\beta$
 is a normalization constant. 
The degree of approximation of a Gibbs sampler as a solver is, thus, captured by a single parameter $\beta \ge 0$, tending to an exact solver for $\beta \to \infty$. 

Given the output distribution of the solver, it is in principle possible (but inefficient, see \cref{app:hardness_proof})
to calculate $M^*_\eta$ exactly. The general idea for an efficient strategy, illustrated in \cref{fig:fig1}, is to instead calculate bounds on the probability of three distinct events: observing (i) a feasible point with objective smaller or equal than $E_f$, (ii) a feasible point with objective larger than $E_f$ and (iii) an infeasible point.  
These three bounds can then be combined to determine $M^\ast \geq M^\ast_\eta$ which will be closer to $M_\eta^\ast$ the tighter the bounds are.  

Evaluating the bounds requires the following weight functions (non-normalized densities) depending on the problem instance: 
first, we define the \emph{penalization degeneracy} as 
\begin{equation}
    \nv(v) = |\{\vec x \in \{0,1\}^n : E^{(p)}(\vec x) = v\}|.  
\end{equation}
This represents the number of bitstrings with penalty function value $E^{(p)}(\vec x) = v$ and enables control over the distribution of infeasible points. 
For many problems, the penalization degeneracy can be obtained  analytically. 
We derive the expressions for  \ac{MNPP}, \ac{TSP} and \ac{PO} in \cref{app:penalization_degeneracies} for small $v$.  
We have found that in practice it is sufficient to evaluate $n_\text{\rm pen}(v)$ only for $v \leq  v_\mathrm{cut}$ up-to some small constant $v_\mathrm{cut}$.
Alternatively, one can resort to a  coarse sampling of the penalization energies of bitstrings and subsequent fit.

Second, we  make use of a lower bound $\ELB \leq E^{(o)}(\vec x)$ on the unconstrained objective function for all $\vec x$. 
Such a bound can be efficiently computed using an SDP relaxation.

Third, we introduce the feasible spectral weights 
\begin{equation} \label{eq:feasible_spectrum}
    n_\Delta(e) = |\{\vec x \in \mathcal F : e \leq E^{(o)}(\vec x) < e + \Delta \}|.
\end{equation}
This can be approximately estimated by randomly sampling a number $N_s$ of bit-strings from a uniform distribution over the feasible set $\mathcal F$ and counting the sampled strings $\vec x$ for which the objective energy $E^{(o)}(\vec x)$ lies within the considered range; rescaled by $|\mathcal F|/N_s$, this count gives the estimate $\hat n_\Delta(e)$.
In practice, it is sufficient to estimate $n_\Delta(e)$ only for $e \in \Lambda = \{0, \Delta, 2\Delta, \dots, \left\lceil(E_\mathrm{max} - \ELB)/\Delta\right\rceil \Delta\}$, where $E_\mathrm{max}$ is the maximal energy sampled.
For structured problems like \ac{TSP}, \ac{MNPP} or \ac{PO}, the sampling is efficient, see \cref{app:subroutines}.

\cref{alg:M} combines these estimates to compute the bounds $\Bfl$, $\Bfh$, and $\Binf$, that bound the probabilities of observing feasible points with a low objective value, feasible points with a high objective value, and infeasible points, respectively. 
To account for the statistical error of $\hat n_\Delta$, the algorithm takes as additional input lower and upper margins $\epsl(e), \epsh(e) \geq 0$ to calculate the lower bound $\Bfl$ and the upper bound $\Bfh$, respectively.  
From $\Bfl$, $\Bfh$, and $\Binf$ we obtain an estimate $M^\ast$ for $M$ in the last step. 
The correctness of the algorithm in the limit of infinite samples is established by the following theorem.

\begin{algorithm}[tb]
\setstretch{1.35}
\caption{
$M(E^{(o)}, E^{(p)}, E_f, \beta, \eta, \vcut, N_s, \Delta, \epsilon)$}
\label{alg:M}
\SetKwInOut{Input}{Input}

\Input{
$E^{(o)}$ (objective),
$E^{(p)}$ (penalty),
$E_f$ (energy threshold),
$\beta$ (inverse temperature), 
$\eta$ (success probability),
$N_s$ (sample size),
$\vcut$ (degeneracy cut-off),
$\Delta$ (energy resolution), and
$\epsilon = (\epsl, \epsh)$ (statistical margins).
}
\vspace{0.5em}
\hrule
\vspace{0.5em}

Determine $E_\mathrm{LB}$ from SDP relaxation of  $\operatorname{minimize}_{\vec x \in \{0,1\}^n} E^{(o)}(\vec x)$ \label{alg:line:ELB} 

Estimate $\hat n_\Delta(e + \ELB)$ for each $e \in \Lambda$ from $N_s$ uniform samples from $\mathcal F$ \label{alg:line:Usample} 

Calculate $\Bfl \coloneqq \sum_{e \in \LambdaLow} \e^{-\beta (e+\Delta)} [\hat n_\Delta(e + \ELB) - \epsl(e)]$ with 
$\LambdaLow =\{0, \Delta, \ldots, (\left\lfloor(E_f - \ELB)/\Delta\right\rfloor-1) \Delta\} \subset \Lambda$ \label{alg:line:Bf_<}

Calculate $\Bfh \coloneqq \sum_{e \in \LambdaHigh} \e^{-\beta e} [\hat n_\Delta(e + \ELB) + \epsh(e)]$ for 
$\LambdaHigh \coloneqq \Lambda\setminus\LambdaLow$
\label{alg:line:Bf_>}

Compute the penalization degeneracy $n_{\mathrm{pen}}(v)$ for $v \in \{0, \ldots, v_\text{cut}\}$\label{alg:line:pen_deg}

Set $\Binf(M) \coloneqq \sum_{v = 1}^{\vcut} \e^{-\beta M v} \nv(v)\,$ \label{alg:line:Binf}

Set $g(M) \coloneqq \Binf(M) + \Bfh - \tfrac{1 - \eta}\eta \Bfl$ \label{alg:line:g}

Determine $M^\ast$ as the root of $g(M)$. \label{alg:line:g_root}

\textbf{Return} $\max\{0,M^*\}$ or $\{\}$ if no roots were found.
\end{algorithm}

\begin{theorem}\label{thm:guarantee}
    In the limit $N_s \to \infty$, for $\epsilon = 0$ and $\vcut = \max_{\vec x} E^{(p)}(\vec x)$, the following holds:
    If \cref{alg:M} returns a solution $M^\ast \neq \{\}$, 
    then \eqref{QUBO} with $M = M^\ast$ is an $\eta$-reformulation with guaranteed energy threshold $E_f$ for a Gibbs sampler at inverse temperature $\beta$.
\end{theorem}

A choice of margins $\varepsilon$ based on a Clopper-Pearson bound \cite{Clopper1934} guarantees correctness of \cref{alg:M} for any finite number of samples $N_s$ with high probability.  We derive this generalization of  Theorem~\ref{thm:guarantee} in Appendix~\ref{app:guarantee_realistic}.

\begin{proof}
The proof first establishes that the three bounds evaluated in the algorithm actually bound the corresponding events, and finally that they are combined into a bound of $M$. 

By the theorem's assumption, $\vec x \in \{0,1\}^n$ is sampled according to the Gibbs distribution $p_\beta(\vec x) = \mathcal N_\beta \e^{-\beta E(\vec x)}$ with normalization $\mathcal N^{-1}_\beta = \sum_{\vec x}\e^{-\beta E(x)}$. 
By $\overline{\mathcal F}\coloneqq \Set{0,1}^n\setminus \mathcal F$ we denote the complement of $\mathcal F$. 
We consider the different values $v\in \ZZ_+$ the penalty term $E^{(p)}$ can take and decompose $\overline{\mathcal F}$ into the preimages of $v\neq 0$ as $\overline{\mathcal F} = \bigcup_{v=1}^\infty (E^{(p)})^{-1}(\Set v)$. 
This observation allows us to use the lower bound $\ELB \leq E^{(o)}(\vec x)$ from step~\ref{alg:line:ELB} of the algorithm to bound the probability of observing an infeasible solution as 
\begin{equation}\label{eq:Binf}
\begin{split}
    \Pr[\einf] &= \sum_{x \in \overline{\mathcal F}} p_\beta(\vec x) 
    = \mathcal N_\beta \sum_{v = 1}^{\infty} 
    \sum_{\vec x \in (E^{(p)})^{-1}(\Set v)} \e^{-\beta (E^{(o)}(\vec x) + M v)} \\
    &\leq 
    \mathcal N_\beta \e^{-\beta \ELB} \sum_{v = 1}^{\infty} \e^{-\beta M v} \sum_{\vec x \in (E^{(p)})^{-1}(\Set v)} 1  \\
    &=
    \mathcal N_\beta \e^{-\beta \ELB} \sum_{v = 1}^{\infty} \e^{-\beta M v} n_\mathrm{pen}(v) = c \,\Binf(M), 
\end{split}
\end{equation}
where we defined the positive constant $c = \mathcal N_\beta e^{-\beta \ELB} > 0$
and $\Binf(M)$ from step~\ref{alg:line:Binf}.

Next, we show that $\Bfl$ from step~\ref{alg:line:Bf_<} is a lower bound on the probability of observing feasible points with low objective. 
In the limit of $N_s\to \infty$, our estimate $\hat n_\Delta$ for $n_\Delta$ is exact, so that for $\epsilon = 0$ the algorithm evaluates $\Bfl$ and $\Bfh$ with the exact weights $n_\Delta$. We divide the relevant energy interval in bins of size $\Delta$, with steps $\LambdaLow = \{0, \Delta, \ldots, (\left\lfloor(E_f - \ELB)/\Delta\right\rfloor-1) \Delta\}$. 
These are the bins lying entirely below $E_f$; for $\Delta=(E_f-\ELB)/K$ with $K\in\mathbb N$, they cover $[\ELB,E_f)$ exactly.

We denote the set of feasible states in these bins by $b(e) = \{\vec x \in \mathcal F : E^{(o)}(\vec x) - \ELB \in [e, e + \Delta)\}$, 
with cardinality $|b(e)| = n_\Delta(e + \ELB)$.
Since $E^{(p)}(\vec x) = 0$ for any $\vec x \in \mathcal F$, we can write the probability from \cref{eq:eta_reformulation} as
\begin{equation}\label{eq:Bfl}
\begin{split}
    p&\coloneqq
    \Pr[\efl] \\    
    &\ge
    \mathcal N_\beta 
    \sum_{e \in \LambdaLow} \sum_{\vec x \in b(e)} \e^{-\beta E^{(o)}(\vec x)} \\
    &\geq \mathcal N_\beta \e^{-\beta \ELB}\sum_{e \in \LambdaLow} \e^{-\beta (e+\Delta)} n_\Delta(e + \ELB) = c \Bfl\,.
\end{split}
\end{equation}

Similarly, for the feasible events with a high objective value, defining $\Bfh$ as in line \ref{alg:line:Bf_>}, we ensure that
\begin{equation}\label{eq:Bfh}
\begin{split}
    &\Pr[\efh] \\
    &\le 
    \mathcal N_\beta 
    \sum_{e \in \LambdaHigh} \sum_{\vec x \in b(e)} \e^{-\beta E^{(o)}(\vec x)} \\
    &\leq \mathcal N_\beta \e^{-\beta \ELB}\sum_{e \in \LambdaHigh} \e^{-\beta e} n_\Delta(e + \ELB) = c \Bfh\,.
\end{split}
\end{equation}

The events are mutually exclusive and complete. Hence, 
\begin{equation}
    \Pr[\einf] + \Pr[\efh] + \Pr[\efl] = 1, 
\end{equation}
which in terms of the bounds \eqref{eq:Binf} and \eqref{eq:Bfh} implies that 
\begin{equation}
\begin{split} \label{eq:Bfh+Binf}
    c\,(\Binf(M) + \Bfh) 
    \geq 
    1- \Pr[\efl]
    = 1 - p\,.
\end{split}
\end{equation}

If step~\ref{alg:line:g} yields $g$ that has a positive root $M^\ast$, then 
\begin{equation} \label{eq:compound_inequality_proof}
\begin{split}
    0 &= \Binf(M^\ast)  + \Bfh - \frac{1-\eta}\eta \Bfl \\
    &\geq  c^{-1}\biggl(1 - p - \frac{1- \eta}\eta p \biggr) =  c^{-1}\biggl( 1 - \frac p\eta\biggr)\,,  
\end{split}
\end{equation}
and thus, with $M=M^\ast$, we have $p \geq \eta$. 

Finally, if $g$ has a negative root, then $p \geq \eta$ already holds for $M^\ast = 0$. 
\end{proof}

A few comments on the algorithm are in order: 

(i) As we argue in \cref{app:eta_exist}, if no $\eta$-formulation with guaranteed threshold $E_f$ exist, the algorithm returns $\{\}$. 
In this case, we can choose any permissible $\eta < \eta_\mathrm{exist} = (1 + \Bfh/\Bfl)^{-1}$ and rerun the algorithm, while ensuring the existence of a solution.

(ii) One can set $E_f = \infty$ (and truncate $\LambdaLow$ at the last energy sampled in step \ref{alg:line:Usample} to not require any guarantee on the objective of the solutions. In this regime $\Bfh = 0$ and Theorem~\ref{thm:guarantee} continues to hold. The algorithm then effectively targets feasibility-only sampling. 

(iii) We numerically observe that in our problem instances, the penalization degeneracy does not grow exponentially, see \cref{app:penalization_degeneracies}. 
Thus, summands with larger $v$ entering $\Binf$ are eventually exponentially suppressed.  
This allows one to use a small value for $v_\text{\rm cut}$ in practice with negligible error: we choose it such that the output $M^\ast$ no longer changes when $v_\text{\rm cut}$ is increased further, see \cref{app:output_dependence}.

Let us now analyze the time and memory complexity of the algorithm in more detail. 
In particular, we establish in the following that 
\emph{the algorithm is efficient for problem instances with 
(i) polynomially bounded entries for $Q$, $A$, and $b$ in the problem specification \eqref{lcBQP}, 
(ii) efficient uniform random sampling from the feasible set and 
(iii) efficient evaluation of $n_\text{\rm pen}(v)$.}  
To control the algorithm's complexity, we impose assumptions on the magnitude of the objective function $E^{(o)}$ and the penalization term $E^{(p)}$.  
Note, that if $Q_{ij}, A_{ij}, b_j \in \mathcal O(1)$ for all $i,j$, then  $E^{(o)} (\vec x) = \vec x^t Q \vec x\le n^2 \norm{Q}_\infty \in \mathcal O (n^2)$ for all $\vec x \in \{0,1\}^n$ and the penalization term $E^{(p)} (\vec x) = \vec x^t A^t A \vec x - 2 \vec x^t A b + b^tb$ with $m$ constraints scales as  $\mathcal O(mn^2)$ dominated by its first summand. 
The examples in our work (MNPP, TSP and PO) indeed all have entries in $A$ and $b$ of constant magnitude independent of the system size. 
More generally, we also directly conclude that if the entries of $Q$, $A$, and $b$ grow at most polynomially with the system size, 
$E^{(o)} \in \mathcal O(\operatorname{poly}(n))$ and $E^{(p)} \in \mathcal O(\operatorname{poly}(n,m))$. 

The maximum penalization value $v_\mathrm{max} = \max_{\vec x} E^{(p)}(\vec x)$ will thus scale like $mn^2$ for the benchmarked problems and other bounded-entries similar problems, or will be in $\poly(m,n)$ in general.

Let us discuss each step of \cref{alg:M}: 

The SDP in \emph{Step~\ref{alg:line:ELB} } generally takes $O(n^6)$ time and $O(n^4)$ memory complexity \cite{nesterov1994interior}.  
Considerable speed-ups can potentially be achieved using sketching techniques \cite{yurtsever2017sketchy}.

\begin{figure*}[htbp]
    \centering
        \includegraphics[width=\linewidth]{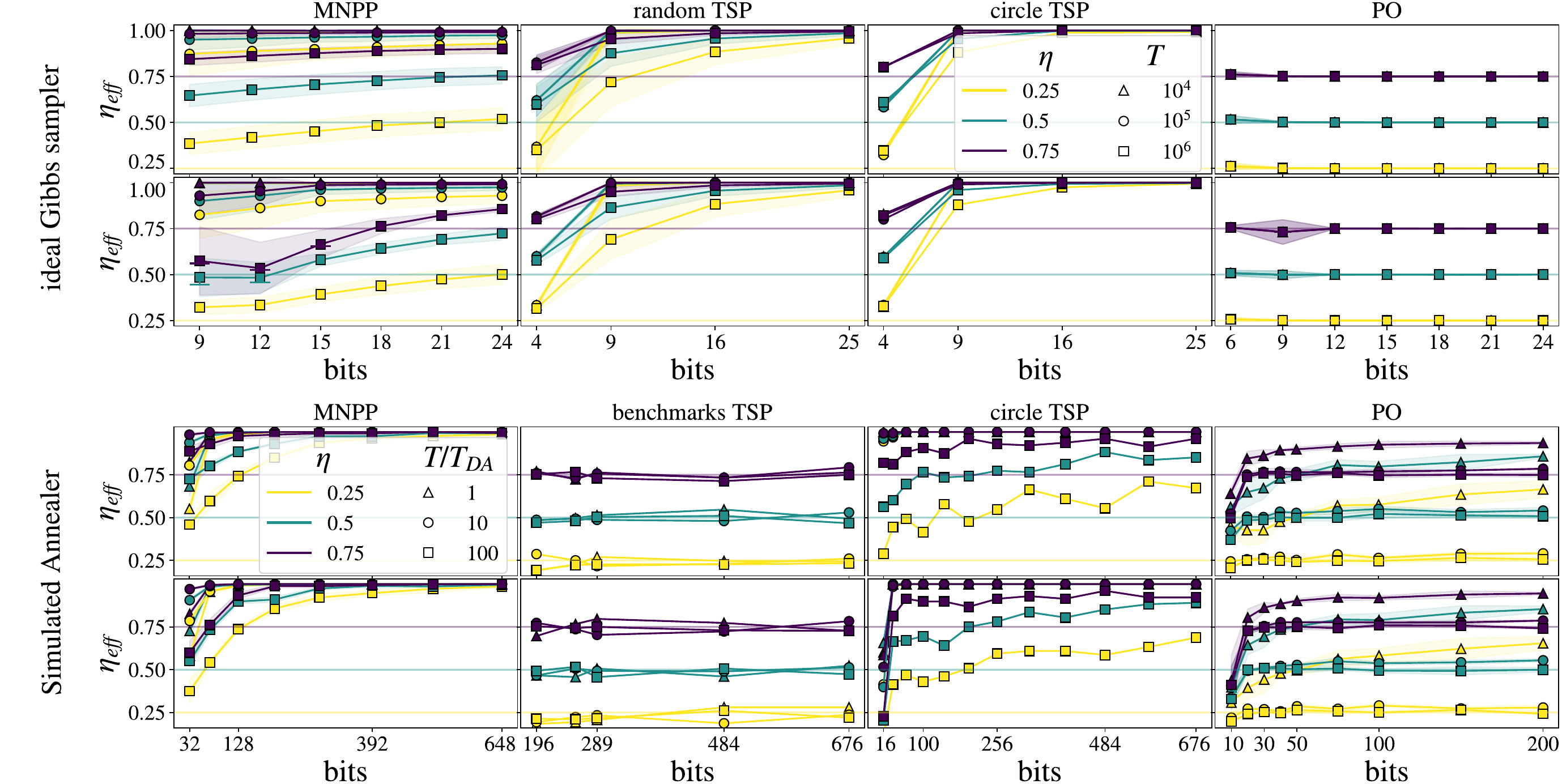}
    \caption{\emph{Effective success probability $\eta_\mathrm{eff}$ of an ideal Gibbs sampler (top rows) and \acf{SA} (bottom rows)} for sampling feasible points, using a \ac{QUBO} reformulation with penalization weight $M^\ast$ calculated by \cref{alg:M} for  
    different constrained optimization problems (columns, see \cref{app:benchmark_problems} for details) and system sizes. 
    For each solver, in the first row we only require feasibility ($E_f=\infty$), while for the second row we further require solutions with objective  smaller than a finite, problem-dependent $E_f$.
    Different colors denote target success probabilities $\eta \in \{0.25,0.5,0.75\}$, with horizontal reference lines at these values. Marker shapes indicate sampler temperature $T=\beta^{-1}$. For SA, temperatures are obtained by rescaling Digital Annealer schedules as $T=\phi T_{\mathrm{DA}}$, with $\phi\in\{1,10,100\}$; for \ac{PO}, schedules are approximated using instances of the same size from other benchmarks. Solid lines and markers show averages over $100$ (ideal Gibbs sampler) or $4$ (\acs{SA}) instances, with shaded standard deviation; \emph{circle \ac{TSP}} and \emph{benchmark \ac{TSP}}, only define a single instance per system size. Per instance, $10^3$ (ideal Gibbs) or $128$ (\acs{SA}) samples are drawn.  \ac{PO} uses $10^5$. 
    We generally observe that $\eta_\text{\rm eff}$ is larger than $\eta$ showing that \cref{alg:M} yields admissible $\eta$-reformulations.  
    For finite $E_f$, some combinations of $T$, $\eta$, and $E_f$ make the target $\eta$ unattainable (see \cref{app:eta_exist}); in these cases, $\eta$ is reduced. Such instances are indicated in the optimality-focused \ac{MNPP} panels (bottom left) by short horizontal bars marking the reduced target below the achieved $\eta_\mathrm{eff}$.
    }
    \label{fig:eta_eff_GSA}
\end{figure*}

For \emph{Step~\ref{alg:line:Usample} } to be efficient, we need to choose $N_s$ to be at most polynomially in $m,n$. 
This will, in turn, introduce a statistical error in the estimation of $n_\Delta$ and, thus, $\Bfl$ and $\Bfh$. 
The statistical error can be accounted for by the margins $\epsilon$ of \cref{alg:M}. 
More precisely, extending Theorem~\ref{thm:guarantee}, we show in \cref{app:guarantee_realistic} that for margins derived from per-bin confidence bounds, \cref{alg:M} yields an $\eta$-reformulation with probability at least $1-\delta$ for every sample size $N_s$ and every resolution $\Delta$. 
Correctness thus imposes no lower bound on $N_s$, and the sample size only determines how conservative the output is. 
Energy bins that contain a fraction of the feasible points much larger than $\ln(|\Lambda|/\delta)/N_s$ are estimated with a small relative error, whereas rarer ones are effectively dropped from $\Bfl$ and overestimated in $\Bfh$. 
This increases the returned $M^\ast$ only logarithmically (Proposition~\ref{prop:finite_sampling_cost}) and, for finite $E_f$, lowers the threshold $\eta_\mathrm{exist}$ of comment (i).
We refer to the appendix for more details. 
Obtaining a single uniform sample from the feasible set is efficient for many problems. 
In particular, for TSP, MNPP and PO it takes $\mathcal O(n)$, $\mathcal O(n)$ and $\mathcal O(n^2)$ time, respectively ($\mathcal O(n_v^2)$, $\mathcal O(N+P)$ and $\mathcal O(N^2)$, in the individual problem parameters, see \cref{app:benchmark_problems}). In all three cases, the memory requirements are of the same order or smaller than the corresponding time complexity. Thus, for these instances and a polynomial number of samples $N_s$, Step~\ref{alg:line:Usample} is efficient. 

When the objective function is polynomially bounded and the resolution is at least inverse polynomial, $\Delta \in \Omega(1/\poly(n,m))$, also the cardinality of the lattice $\Lambda$, which contains $\lceil (E_\mathrm{max} - \ELB)/\Delta \rceil + 1$ points, is at most polynomial.  
Thus, the summations in \emph{Steps \ref{alg:line:Bf_<} and \ref{alg:line:Bf_>}} are efficient. 

\emph{Step~\ref{alg:line:pen_deg}} computes the penalization degeneracy. 
Even setting $v_\text{\rm cut} \to \infty$, this involves evaluating $n_\text{\rm pen}(v)$ for $v \in \{1, \ldots, v_\mathrm{max} = \max_{\vec x} E^{(p)}(\vec x)\}$. 
For $Q$, $A$, and $b$ with polynomially bounded entries in $n,m$, we thus compute $n_\text{\rm pen}$ for polynomially many arguments.  
For instances where a single evaluation of $n_\text{\rm pen}$ is efficient, the step is, thus, overall efficient. 

\emph{Step~\ref{alg:line:Binf}} is efficient under the same assumption as Step~\ref{alg:line:pen_deg}. 

\emph{Step~\ref{alg:line:g_root}} requires constant evaluation of $\Binf$ (Step~\ref{alg:line:Binf}) to determine an integer approximation to a root of $g$ \cite{burden2015numerical}.

Thus, we conclude that under assumptions that are often met by problems under consideration, one can choose the parameters of the algorithm such that it is efficient and guaranteed to yield the targeted reformulation with high probability.  
In particular, for MNPP, TSP, and PO, we have $O(1)$ entries of $A$ and $b$ and $v_\mathrm{max}$ will scale like $mn^2$, which simplifies the overall complexity analysis.

Putting everything together and under reasonable assumptions, the dominant contributions to the total computational cost arise from the SDP relaxation, scaling as $\mathcal{O}(n^6)$ and from uniformly sampling the feasible set, scaling as $\mathcal{O} (N_s n^2)$. Since the guarantee of \cref{app:guarantee_realistic} holds for every sample size, $N_s$ can be treated as a constant or chosen polynomially in $n$. The overall complexity is then polynomial in the system size and, for constant $N_s$, ultimately dominated by the $\mathcal{O}(n^6)$ cost of the SDP relaxation. Finally, note that for many problems (e.g. TSP and MNPP) a trivial lower bound exists ($\ELB = 0$) and thus does not require explicit computation.

\begin{figure}[htbp]
    \centering
        \includegraphics[width=\columnwidth]{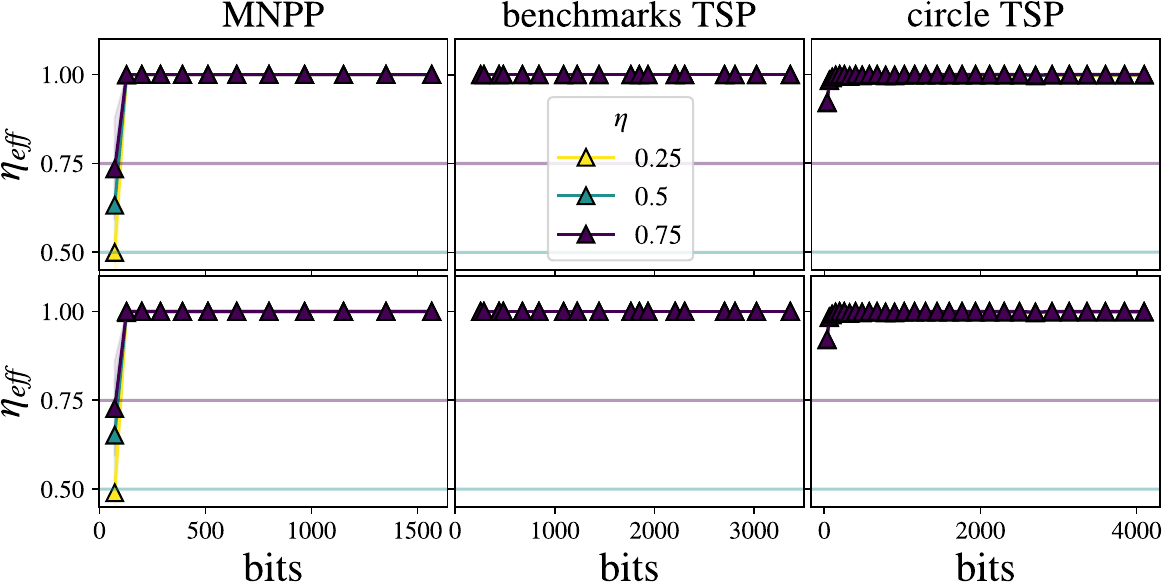} 
    \caption{\emph{Effective success probability $\eta_\mathrm{eff}$ of the Fujitsu Digital Annealer (version 3)} as a function of the system size using penalizations weigths $M^\ast$ determined by \cref{alg:M}. 
    The structure of the figure is identical to \cref{fig:eta_eff_GSA} and refer to its captions for details.
    The temperatures of the annealing process here have been automatically selected internally in the Digital Annealer. 
    For \ac{MNPP}, solid lines and markers are averages over $4$ instances, with shaded standard deviation. For \ac{TSP} and circle \ac{TSP} only one instance was considered per system size. For each instance, $512$ solutions were sampled for all problems.
    }
    \label{fig:eta_eff_DA}
\end{figure}

Note that generally speaking, we are here trading computational effort with the tightness of the bounds on the events and, thus, the degree of approximation of the optimal value of $M$. 

Further results concerning the algorithm are presented in the appendix. 
In particular, a numerically robust modification of the algorithm is described in \cref{app:logspace_implementation}, a brief study of the output $M^*$ dependence on the hyperparameters $E_f$ and $v_\mathrm{cut}$ is reported in \cref{app:output_dependence}, and the applicability of the algorithm to the inverse problem of determining an appropriate solver temperature $\beta$ for a fixed penalty weight $M$ is detailed in \cref{app:inverse_problem}.

\subsection{Validation and numerical benchmarks}
We benchmark the proposed strategy on three classes of constrained problems: \Acf{TSP}, \acf{MNPP} and \acf{PO}. Complete formulations of these optimization problems are given in \cref{app:benchmark_problems}. 
The problems capture distinct  structures of sets of feasible points. 
As solvers we use an ideal Gibbs sampler, a simulated annealing algorithm,
and the Digital Annealer.
All instances are tested following the same scheme: \cref{alg:M} is used to determine  the penalization constant $M$ for $\eta \in\{0.25, 0.5, 0.75\}$ and for three different target temperatures (except for the \ac{DA}, operating at a single, automatically selected temperature). 
Our choices of $N_s$ are sufficiently large to set as an approximation $\epsl = \epsh = 0$ throughout.
We choose $E_f$ such that it is not impeding the success probability—i.e., so that $\eta_\mathrm{exist}$ does not become  small and render the tests uninformative (see \cref{app:eta_exist}). 
We set $E_f = \alpha n^2$, where $n$ is the number of bits in the instance and choose $\alpha$ accordingly or we set $E_f = \infty$, only enforcing feasibility. 
From multiple runs of the solvers we then estimate the effective success probability $\eta_\mathrm{eff}$, i.e.\ the fraction of observed feasible solutions with energy smaller than $E_f$, per problem class and system size. 

In \cref{fig:eta_eff_GSA} top rows, we display the effective success probability  for an ideal Gibbs sampler and different choices of $E_f$, respectively. 
We observe  that, with one exception, $\eta_\mathrm{eff}$ is consistently above the desired threshold. 
The ideal Gibbs sampler exactly  fulfils the assumptions underlying our method development. 
Its output distribution is fully characterized by a known inverse temperature $\beta$ used to determine $M$. %
The gap between $\eta_\mathrm{eff}$ and $\eta$ for \ac{MNPP} and \ac{TSP}, thus, indicates that the bounds evaluate in \cref{alg:M}  are not tight. 
For \ac{PO}, in contrast, we observe that  $\eta_\mathrm{eff}$ is indeed very close to $\eta$. 

For \ac{MNPP} with $E_f = \alpha n^2$, we encounter the only exception where $\eta_\text{eff} < \eta \in \{0.5, 0.75\}$ for $T=\beta^{-1}=10^6$ and small system sizes (\cref{fig:eta_eff_GSA} second row, first plot, squares).
Here, it is indeed impossible to sample feasible solutions \emph{below} the specified target energy $E_f$ with sufficient probability (see \cref{app:eta_exist} for details). 
The algorithm terminated with $\{ \}$ and the points were obtained 
re-running the algorithm with an $\eta < \eta_\mathrm{exist}$. 
The updated $\eta$ is shown as the the 
horizontal bar. 
We observe that $\eta_\text{eff}$ is consistently above this threshold, confirming the robustness of the algorithm also in this setting.

The lower two rows of \cref{fig:eta_eff_GSA} show the effective success probability $\eta_\mathrm{eff}$ for simulated annealing. 
The simulated annealing uses the same   temperature schedule as the Digital Annealer, when available; for \ac{PO}, the schedule from other instances of the same size was used. 
We mostly observe comparable results to the Gibbs sampler. 
For high temperatures $\eta_\mathrm{eff}$ is found slightly below $\eta$ for \ac{TSP} and \ac{PO}. 
The solver is not guaranteed to have converged to a stationary Gibbs distribution at the end of the annealing and  potentially deviates from our idealized assumptions underlying our method.
Nonetheless, our algorithm determines values of $M$ with a comparable success probability $\eta_\mathrm{eff}$ to the targeted one in this setting. 

For \ac{DA} (version 3) 
we find that $\eta_\text{\rm eff}$ is consistently higher than the targeted $\eta$ in all settings,  \cref{fig:eta_eff_DA}. 
Moreover, all $\eta_\text{\rm eff}$ quickly converges to $1$ in the system sizes. This indicates that our methods  overestimates $M$ for larger system sizes.  
We have already seen in \cref{fig:bigM_prob} that for large systems a slight overestimation of $M$ yields quickly to nearly only observing feasible solutions.  
The output distribution of the \ac{DA} is only qualitatively approximated by a Gibbs distribution with inverse temperature $\beta$ determined in the annealing schedule.  
Our results show the \ac{DA} is further biased towards low-objective values than the Gibbs distribution. 
Thus, our methods determines practical values for $M$ for Fujitsu's Digital Annealer Unit on instances with over thousand variables ensuring reliable performance of the solver.  

\begin{figure*}[htbp]
    \centering
        \includegraphics[width=\linewidth]{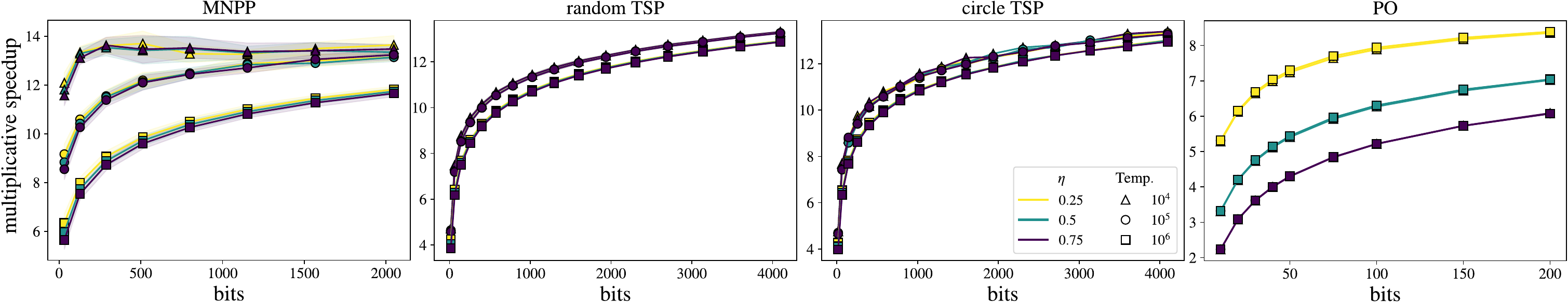}
    \caption{\label{fig:iter_saved}
    \emph{Multiplicative speedup compared to binary search from direct bound,} 
     computed as $\log_2( M_{\ell_1} / M^*)$, 
    as a function of the system size for the different benchmarked problems. Here $M^*$ is the output of \cref{alg:M} and $M_{\ell_1}$ is a direct bound for $M$ (see \cref{lem:Ml1}). 
    Colors indicate different target probabilities $\eta \in \{0.25, 0.5, 0.75\}$, and marker shapes different temperatures $T = \beta^{-1}$. For \ac{MNPP}, random \ac{TSP} and \ac{PO}, lines and markers are averages over $10$ instances and the standard deviation is shaded. Circle \ac{TSP} defines one instance per system size. 
    As shown in \cref{fig:bigM_prob}, directly using $M_{\ell_1}$ substantially degrades the solution quality. 
    A binary search reducing $M_{\ell_1}$ to $M^\ast$ requires iterations $\log_2( M_{\ell_1} / M^*)$ with repeated class to the \ac{QUBO} solver.
    Thus, the advantage of using $M^*$ is a reduction in overhead proportional to this factor.
    }
\end{figure*}

A common approach in practice is to perform a brute-force search for a suitable value for $M$ \cite{Garcia_penalty_22, Roch_CE_21}, e.g.\ using binary search.
Starting from an overestimated value for $M$, the penalization is iteratively halved, and the problem is solved until an unacceptable number of infeasible solutions start appearing.
This method traverses the search space exponentially fast. 
The overall cost of the method, however, depends on the complexity of the individual solver calls.  Depending on the solver, its hardware platform and the problem sizes, the costs can be prohibitive already for a small number of iterations.  
This still motivates using a good estimate for $M$ as the initial value of the binary search, as provided by our method developed here.  
So the practical benefit of the method can be quantified as the number of solver calls that are `saved' when initializing binary search using \cref{alg:M} instead of with a more direct upper bound. 
As we show in \cref{app:M_l1}, a simple efficiently computable $M$ yielding an $\eta$-reformulation of a Gibbs sampler at inverse temperature $\beta$ is given by 
\begin{equation} \label{eq:M_l1}
    M_{\ell_1}(\beta) = \beta^{-1}(n \ln 2 -\ln (1-\eta)) + \norm{Q}_{\ell_1}\,. 
\end{equation}
Note that this definition makes use of the direct upper bound $\norm{Q}_{\ell_1}$ on the objective function that is often used to estimate $M$ for an exact solver \cite{qiskit_doc, Garcia_penalty_22}.
The number of `saved' solver calls is, thus, computed as $\log_2(M_{\ell_1} / M^\ast)$ and depicted for the benchmarking instances in \cref{fig:iter_saved}.
We observe that across all temperatures studied, our algorithm results in reductions of the number of solver calls, often reducing the runtime by factors of $10$ or more. 
This shows that our method can yield practical advantages in the time to solution.

{
\subsection{Application to QAOA} \label{sec:qaoa} 

An important application for our algorithm are combinatorial quantum solvers. Many quantum solver paradigms require a QUBO formulation. Due to the high cost of individual quantum solver calls on hardware, careful classical pre-processing is justified. Our method can be straight-forwardly extended to quantum solvers exhibiting approximate thermal output distributions. As an example, we consider the Quantum Approximate Optimization Algorithm (QAOA) \cite{Farhi_QAOA_14,Blekos_QAOA_23}, a variational circuit of depth $p$ alternating evolutions under a Hamiltonian encoding the QUBO energy $E(\mathbf{x})$, and a so-called transverse-field mixer $H_{\mathrm{mix}} = \sum_i X_i$.  
We deploy our algorithm using effective temperature heuristics that model thermal QAOA outputs and benchmark the solver performance in classical simulations of the quantum circuits and in a trapped-ion quantum processor \cite{h2_system}. 

Single-layer \ac{QAOA} was shown to produce thermal-like, pseudo-Boltzmann states \cite{DiezValle_pseudoBoltz_23}, and approximately Boltzmann output statistics were observed numerically for MaxCut at larger depths \cite{Lotshaw_Boltzmann_23}.
Recent work \cite{DiezValle_Teff_25} studied QAOA circuits with optimized angles for dense random QUBOs, and found numerically that their output distributions are well described by \emph{bimodal pseudo-Boltzmann distributions}:
a dominant cold component at an inverse temperature that grows linearly with the circuit depth $p$, and a secondary, hotter component, whose weight is suppressed with increasing depth. 
This result motivates the model of a depth-$p$ QAOA circuit being a Gibbs sampler at the effective inverse temperature
\begin{equation}
    \beta_p = \kappa\, p\, n^{-3/2},
    \label{eq:beta_qaoa}
\end{equation}
where the constant $\kappa$ is fixed empirically from the data in Ref.~\cite{DiezValle_Teff_25}. 
Note that our model entails an approximation and an assumption: we operate under the dominant cold mode restricting the pseudo-Boltzmann distribution to a single thermal component, and we assume that the scaling law \eqref{eq:beta_qaoa} (established in Ref.~\cite{DiezValle_Teff_25} on dense random QUBOs) transfers to the structured, constrained problems considered here. These premises allow our method to carry over unchanged.

A subtlety arises from different normalization of energies and inverse temperatures.
The scaling law \eqref{eq:beta_qaoa} is extracted in Ref.~\cite{DiezValle_Teff_25} from instances with normalized energies in $[0,1]$.
In contrast, \cref{alg:M} does not require normalization of the energy scale. Thus, we have to rescale inverse temperatures by a suitable energy scale. 
We choose $s_{\text{conv}} = E_{\mathrm{UB}} - E_{\mathrm{LB}}$, an upper bound on the span of the objective, and set $\beta = \beta_p/s_{\text{conv}}$ as \cref{alg:M} input. Here, $E_{\mathrm{LB}}$ is the bound already computed in Step~1 of \cref{alg:M}, and $E_{\mathrm{UB}}$ is obtained from the same SDP relaxation applied to $-E^{(o)}$, at the same negligible cost.
This is an estimate of the relevant energy scale: the penalization exists to suppress the infeasible states, so a sampler at the working temperature is largely confined to the feasible set. The energies it must resolve are therefore the feasible ones, spanning $[E_{\mathrm{LB}}, E_{\mathrm{UB}}]$.

Before compiling the QUBO into rotation gates, we further normalize it by $s_{\text{compile}} = \lVert Q(M^*) \rVert_\infty$, where $Q(M)$ is the resulting QUBO matrix in \cref{QUBO} encoding both objective and penalization contributions,to control the range of the angles and guarantee reliable training. This normalization also affects the angle parameters of the quantum circuit such that the ideal output distribution is unaltered, not affecting the choice of $M^\ast$.

We use a standard QAOA ansatz for all problem instances except for \ac{PO}, where a single pair of angles per layer is too restrictive. There we employ the multi-angle QAOA (maQAOA) \cite{Herrman2022_maqaoa}, which assigns an independent angle $\gamma$ to each Pauli term of the cost Hamiltonian and an independent angle $\beta$ to each qubit of the mixer. 
The cost Hamiltonian of an $n$ variable QUBO carries at most $n(n+1)/2$ Pauli terms, so a depth-$p$ maQAOA circuit is parametrized by a number of independent angles $n_{\text{ang}} \le p\,\Big(\frac{n(n+1)}{2} + n\Big)$, compared to the $2p$ of standard QAOA. The bound is saturated for dense QUBOs. 
To keep the optimization tractable we initialize it with the INTERP heuristic \cite{Zhou_interp_20}, which builds the angle schedule at depth $p+1$ by interpolating the optimized schedule at depth $p$.

The penalization weight $M^*$ and Eq.~\eqref{QUBO} are unchanged. For maQAOA, we use the same scaling law \cref{eq:beta_qaoa}, although it was fitted to QAOA, and validate this choice a-posteriori through the performance of the resulting $M^*$.
 
We benchmark our approach on the same problem classes as before, first using exact state-vector simulation of the noiseless circuits and focusing on feasibility alone ($E_f = \infty$). The variational angles are determined with the Adam optimizer, initialized from a linear-ramp schedule perturbed by random noise, except for the \ac{PO} runs, which use maQAOA with INTERP initialization. The penalization weights $M^*$ are computed for target success probabilities $\eta \in \{0.25, 0.5, 0.75\}$, at different depths $p$, and sizes $n$. Per problem class and size, $10$ random instances are simulated. 

Fig.~\ref{fig:eta_eff_QAOA} (left and center panels) shows the resulting
effective success probabilities for \ac{MNPP} and \ac{TSP}.
In almost every case, $\eta_{\text{eff}}$ satisfies the requirement $\eta_{\text{eff}} \geq \eta$ and it saturates close to one. The exception is \ac{TSP} at $p=5$.
There, $\eta_{\text{eff}}$ falls short of saturation and degrades as the system grows, until it violates the requirement at $\eta = 0.75$.

Two effects plausibly contribute. At fixed $p$ the circuit complexity is bounded, limiting the thermalization of the output distribution. As the spectrum grows, the approximation as a  pseudo-Boltzmann distribution becomes less accurate. 
Moreover, our working assumption neglects the higher temperature component identified in Ref.~\cite{DiezValle_Teff_25}, which carries more weight at small $p$, where we indeed observe a deviation.

The \ac{PO} instances behave differently. With the standard ansatz, QAOA fails to reach the targeted success probabilities already at the smallest sizes. The multi-angle ansatz removes the obstruction: Fig.~\ref{fig:eta_eff_QAOA} (right panel) shows that maQAOA exceeds the targeted $\eta$ already at depth $p = 2$, the target being missed at $p = 1$ only for the larger sizes and at the most demanding requirement, $\eta = 0.75$.

\begin{figure}[htbp]
    \centering
        \includegraphics[width=\columnwidth]{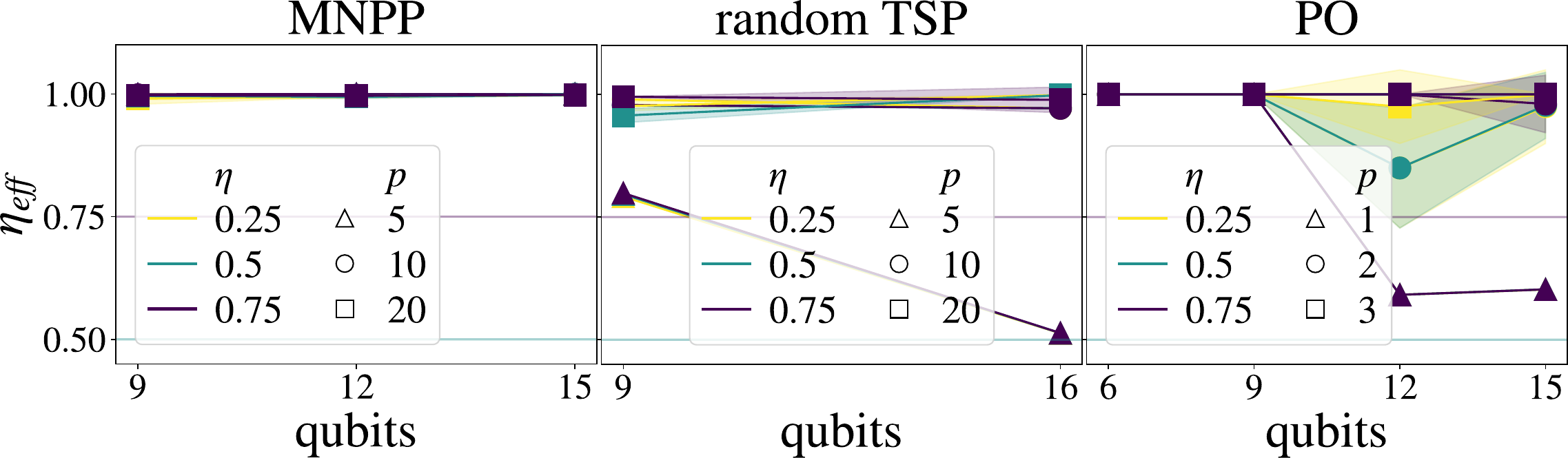} 
    \caption{
    \emph{Effective success probability $\eta_\mathrm{eff}$ of a QAOA circuit} as a function of the system size using penalizations weights $M^\ast$ determined by \cref{alg:M}.
    Figure is similar to \cref{fig:eta_eff_GSA} and we refer to its captions for details;
    marker shapes here encode the circuit depth $p$, instead of the temperature.  
    Solid lines and shaded bands are the mean and standard deviation over $10$ random instances per size. For \ac{PO} only, the circuit implements the multi-angle QAOA ansatz rather than the standard one. 
    }
    \label{fig:eta_eff_QAOA}
\end{figure}

\paragraph*{Quantum hardware deployment}
We complement the simulations with a deployment on Quantinuum's H2-2 quantum processor \cite{h2_system}, a trapped-ion device based on the QCCD architecture in which ions are physically transported between dedicated gating zones. Its all-to-all connectivity is particularly well suited to the dense QUBOs of \eqref{QUBO}, whose two-qubit couplings can be realized without any SWAP overhead, at the cost of a comparatively low circuit repetition rate. 

The device natively implements arbitrary single-qubit rotations,
the two-qubit rotation $R_{ZZ}(\theta) = \exp(-i\theta Z\otimes Z/2)$ with arbitrary angle $\theta$ and virtual $R_Z$ rotations. We compile the trained ansatz to this gate set, through the tket \cite{pytket_paper, pytket_package} compilation pass for the H2 backend, and submit the resulting circuits.

We deploy PO instances at $n=6$ and $15$ using the maQAOA ansatz trained in noiseless simulations with the INTERP strategy for $p=1,2,3$. Circuit dimensions range from $25$ single-qubit and $15$ two-qubit gates, to $211$ single-qubit and $313$ two-qubit gates.

Here we target the quality of the solutions keeping solutions feasible.
We therefore set $E_f$ to the $75$th percentile of the objective energies over the feasible set, so that $\eta_{\mathrm{eff}}$ measures the probability of sampling a feasible string of reasonable quality. We quantify the performance with two metrics. The first one is $\eta_{\mathrm{eff}}$ itself, which assesses the effectiveness of the penalization weight. The second one is the approximation ratio, $ \operatorname{AR}(\vec x) = (E^{(o)}_{\max} - E^{(o)}(\vec x))/(E^{(o)}_{\max} - E^{(o)}_{\min})$, 
where $E^{(o)}_{\min}$ and $E^{(o)}_{\max}$ are the extremal objective energies over the feasible set. For each instance, we average $\operatorname{AR}$ over the shots, so that it measures the typical quality of the states that the circuit samples.

Each instance is run with $M^*$ returned by \cref{alg:M} and with the baseline $M_{\ell_1}$ of Eq.~\eqref{eq:M_l1}, under otherwise identical conditions (same circuit ansatz,  trained angles protocol and effective temperature).

We present the hardware results in Fig.~\ref{fig:QPU_experiment}, obtained on $10$ instances per size with $100$ ($140$) shots per circuit at $n=6$ ($15$). At the smaller size the two weights are indistinguishable (left panel), both with $\eta_{\mathrm{eff}} > 0.95$ at every depth. The separation emerges at $n=15$. For $M^*$, $\eta_{\mathrm{eff}}$ exceeds the target $\eta = 0.5$ at every depth, reaching $\simeq 0.92$ at $p = 2,3$. With $M_{\ell_1}$, instead, it drops to $0.69$ at $p = 2$ and to $0.48$ at $p = 1$, thus failing the requirement. We attribute this degraded performance of $M_{\ell_1}$ to a loss of resolution on the objective, when an unnecessarily large weight dominates the normalized QUBO.

The right panel reports the distribution over instances of the average approximation ratio, together with the value attained by uniform sampling over the feasible set. The advantage of $M^*$ here grows with the depth at both sizes: the two weights are equivalent at $p=1$, where the sampled states barely improve on the uniform baseline, and separate at $p=2,3$, with the median $\operatorname{AR}$ significantly greater than the corresponding value in the $M_{\ell_1}$ case. The comparison is most telling in this configuration $n=15$ and $p=3$: the two weights deliver comparable effective success probabilities, yet clearly distinct $\operatorname{AR}$ distributions. The penalization weight returned by \cref{alg:M} therefore does not merely increase the feasible fraction of the samples, but it concentrates it on the low-energy portion of the feasible spectrum, which is the quantity of interest once feasibility is secured.

Overall, these results indicate that \cref{alg:M} determines practical penalization weights also for quantum solvers.

\begin{figure*}[htbp]
    \centering
        \includegraphics[width=.9\linewidth]{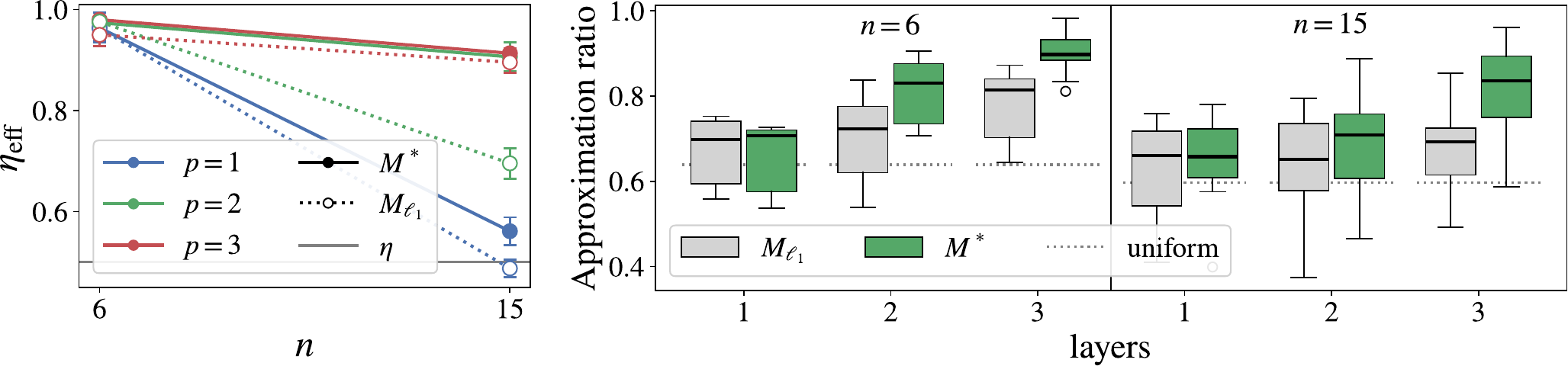} 
    \caption{
    \emph{Experiment on the Quantinuum H2-2 quantum processor.} 
    Portfolio optimization instances of size $n=6$ and $n=15$, $10$ random instances per size, solved with multi-angle QAOA circuits at depths $p=1,2,3$. The variational angles are optimized in noiseless simulation with the INTERP strategy and the resulting circuits are executed on hardware with $100$ ($140$) shots per circuit at $n=6$ ($15$). Each instance is run under two penalization weights, and otherwise identical conditions: the weight $M^*$ returned by \cref{alg:M} (solid lines, green boxes) and the reference weight $M_{\ell_1}$ of Eq.~\eqref{eq:M_l1} (dotted lines, grey boxes). 
    Here the algorithm targets solution quality besides keeping the solutions feasible. 
    The quality threshold $E_f$ is set to the $75$th percentile of the objective energies over the feasible set, and the requested success probability is $\eta = 0.5$. \textit{Left:} effective success probability $\eta_{\mathrm{eff}}$, as a function of the problem size. Points show the mean over the $10$ instances, error bars the standard error of the mean, and the grey horizontal line the target $\eta$. \textit{Right:} distribution over the $10$ instances of the average approximation ratio $\operatorname{AR}$ of the sampled states, normalized to the extremal objective energies of the feasible set. Boxes span the interquartile range with the median as a horizontal line, 
    whiskers follow the $1.5$ inter-quartile range convention and outliers are shown individually; the dotted line is the $\operatorname{AR}$ attained by uniform sampling over the feasible set.
    }
    \label{fig:QPU_experiment}
\end{figure*}

}

\section*{Discussion}

After formally establishing its hardness in general, we introduced an efficient algorithm to determine the penalization weight in unconstrained reformulations of constrained optimization problems, specifically tailored for approximate solvers that are qualitatively similar to  Gibbs samplers. 
More precisely, given a Gibbs sampler at an arbitrary inverse temperature, our algorithm returns a penalization weight such that the solver outputs, with high probability, feasible solutions with objective value below a threshold with a controllable, guaranteed minimum probability.
We demonstrate the practical applicability of our technique beyond exact Gibbs sampling in numerical tests with a simulated annealing algorithm and Fujitsu's Digital Annealer (version 3), for different constrained problem classes and system sizes up to 4096 bits. 
We show that, using our algorithm, one can reduce the time to solution on the solver by an order of magnitude compared to strategies based on binary search for a penalization weight.  

Our algorithm for addressing the Big-$M$ problem improves on state-of-the-art general heuristics for penalization constants. 
It provides a tool to precisely set the penalization constant so as to control the probability of success, using knowledge of the problem structure and the solver’s statistical behavior. 
To this end, we trade resources in the pre-processing for potentially crucial reductions in the solver's runtime.
Further, building on recent work that connects quantum solvers and Gibbs samplers \cite{DiezValle_Teff_25, Benedetti_Teff_16, Oshiyama_Teff_22, Vuffray_Gibbs_22, DiezValle_pseudoBoltz_23, Lotshaw_Boltzmann_23}, we show how our algorithm improves results for \ac{QAOA} for the studied instances. We realized these improvements in an experimental deployment of small problem instances on an ion-trap quantum computer.

Practical methods for addressing the Big-$M$ problem are especially relevant for quantum solvers, e.g. quantum annealers \cite{Kadowaki_QA_98, Farhi_QA_01, Rajak_QA_22, Munoz-Bauz_QA_25} or \ac{QAOA} \cite{DiezValle_Teff_25, Oshiyama_Teff_22, Farhi_QAOA_14, Blekos_QAOA_23}, which typically operate on unconstrained problem formulations. 
Akin to custom hardware solvers like the Fujitsu's Digital Annealer (version 3), classical computational power for the pre-processing is abundant compared to the resource costs of the quantum solver. 
Our method can be adapted to more complex quantum solvers taking into account more specific characterizations of output distributions, complementing existing approaches \cite{Hadfield2019, Herman2023, DiezValle2023, Bucher2025, Bucher2025_MultiConstraint, Shirai2025, Bako2025progqaoaframework, egginger2026_inequalityconstrained, Wang_XY_20, Fuchs_mixers_22}.

\section*{Acknowledgements}
M.~Krispin and M.~Kliesch are funded by 
the Hamburg Quantum Computing project, which is co-financed by the ERDF of the European Union and the Fonds of the Hamburg Ministry of Science, Research, Equalities and Districts (BWFGB); 
and by the Fujitsu Germany GmbH and Dataport as part of the endowed professorship ``Quantum Inspired and Quantum Optimization.''
The hardness proofs in App.~\ref{app:hardness_proof} were AI-assisted based on a human-provided proof sketch. The AI in particular eventually proposed the reduction target. 
Further, a minor inconsistency in treating the finite sampling regime in App.~\ref{app:guarantee_realistic} was identified in an AI-assisted audit.  A first draft of the revised  App.~\ref{app:guarantee_realistic} was generated using AI and reworked by authors.

\input{myacronyms} 
\bibliography{Citations}

\appendix
\crefalias{section}{appendix}
\onecolumngrid
\clearpage

\section{Approximating  \texorpdfstring{$M^*_\eta$}{M*eta} to constant and exponential additive precision is \texorpdfstring{$\mathsf{NP}$}{NP}- and \texorpdfstring{$\#\mathsf{P}$}{\#P}-hard,  respectively} \label{app:hardness_proof}

{
In this appendix we characterize the complexity of computing $M^*_\eta$. The results below follow from polynomial-time Turing reductions from graph problems. 
As we show, deciding whether $M^*_\eta\le t$, or computing $M^*_\eta$ to exponentially small additive precision, is $\#\mathsf{P}$-hard (\cref{thm:sharpP}), and computing $M^*_\eta$ to within a constant additive error is $\mathsf{NP}$-hard (\cref{thm:NP}). 
The underlying reason is that $M^*_\eta$ is fixed by the ratio of partition functions. 
In general, partition functions are hard to compute: their exact evaluation is $\#\mathsf{P}$-hard, while approximating them is hard even at fixed temperature on bounded-degree graphs \cite{Sly_twospin_14}. 
A single ancilla gadget makes this precise (\cref{lem:gadget}): it turns $M^*_\eta$ into an explicit function of the partition function of an arbitrary QUBO.

We begin by recalling and introducing the notation used in this section. 
For $\vec x=(\vec x_1,\dots,\vec x_n)\in \{0,1\}^n$ we write $|\vec x|:=\sum_i \vec x_i$ for the Hamming weight and $\mathbf 0$ for the all-zero string.
We consider Problem~\eqref{lcBQP} with rational $Q$ and feasible set $\mathcal F = \{\vec x \mid A\vec x = \vec b\}$. The penalty function \eqref{eq:penalty_energy} is $0$ when evaluated on $\vec x \in \mathcal F$ and a positive integer otherwise.
For $M\in\mathbb R$, the penalized QUBO is $E(\vec x) = E^{(o)}( \vec x ) + M E^{(p)}( \vec x )$.

At inverse temperature $\beta>0$, the Gibbs sampler outputs $\vec x$ with probability $p_\beta(\vec x) = p_{\beta,M}(\vec x)=e^{-\beta E(\vec x)}/\mathcal Z_{\beta,M}$, where $\mathcal Z_{\beta,M}=\sum_{\vec x}e^{-\beta E(\vec x)} = \mathcal N_\beta^{-1}$. Since $E=E^{(o)}$ on $\mathcal F$, the probability of sampling a feasible point is
$
\Pr[x \in \mathcal F] =p_{\beta,M}(\mathcal F)= \mathcal Z_\mathcal F/\mathcal Z_{\beta,M}
$
with 
$\mathcal Z_\mathcal F:=\sum_{\vec x\in\mathcal F}e^{-\beta  E^{(o)}(\vec x)}$.
For a QUBO matrix $Q$ the partition function is $\mathcal Z_\beta(Q):=\sum_{\vec x\in\{0,1\}^n}e^{-\beta \vec x^\top Q \vec x}$. A constant offset $c$ in the energy multiplies it by the known factor $e^{-\beta c}$, so constants can be omitted.

Hereafter we set $E_f = + \infty$, so that all feasible strings are accepted. We show that the problem remains hard even in this case. 
Recall that we say that
the penalized problem is an \emph{$\eta$-reformulation} if $p_{\beta,M}(\mathcal F)\ge\eta$ and that  
the optimal penalty weight is the minimal value that guarantees an \emph{$\eta$-reformulation}:
\begin{equation}\label{eq:Mstar}
M^*_\eta:=\inf\{M\in\mathbb R:\ p_{\beta,M}(\mathcal F)\ge\eta\}.
\end{equation}

The optimal penalty weight $M^*_\eta$ is in general irrational even for rational inputs. So computing $M^*_\eta$ has to mean deciding inequalities about it or approximating it. 
We formally define the problems: 
\begin{itemize}
\item \textup{\mdseries\textsc{FindM}}: Given $(Q,A,\vec b,\beta,\eta,\delta)$ with $\eta\in(0,1)$ and $\delta>0$, find a rational $\tilde M$ with $|\tilde M-M^*_\eta|\le\delta$.
\item \textup{\mdseries\textsc{FindM-Dec}}: Given $(Q,A,\vec b,\beta,\eta,t)$,  decide if $M^*_\eta\le t$. 
\end{itemize}
We further assume that in both problems the input satisfies $\emptyset\ne\mathcal F\ne \{0,1\}^n$, which makes $M^*_\eta$ a well-defined real number (\cref{lem:mono}).

Let us now formally introduce the source problem used in our reduction.  
Given a graph $G = ([n], E)$, a subset of the nodes $I \subseteq [n]$ is \emph{independent} if for any two vertices $i,j \in I$, $(i,j) \notin E$. 
The graph's \emph{independence number $\alpha(G)$} is the maximum cardinality $|I|$ of an independent set of $G$. We will need the following counting and decision problems:
\begin{itemize}
\item \textup{\mdseries\textsc{CountIndependentSets}}. Given a graph $G$, determine the number of independent sets of $G$.
\item \textup{\mdseries\textsc{MaxIndependentSet}}. Let $G$ be a graph and $k$ an integer, decide if $\alpha(G)\ge k$.
\end{itemize}
Classical results establish that counting the independent sets of a graph is $\#\mathsf{P}$-complete~\cite{Valiant79b,ProvanBall}
and that the decision problem of the independence number is $\mathsf{NP}$-complete \cite{Karp,GareyJohnson}.

We now reduce these graph problems to \textup{\mdseries\textsc{FindM}} and \textup{\mdseries\textsc{FindM-Dec}}. The reductions rely on the following observation: introducing an auxiliary variable makes $M^*_\eta$ an explicit function of $\mathcal Z_\beta(Q_G)$ for any QUBO matrix $Q_G$ using a simple gadget (\cref{lem:gadget}). 
By choosing a suitable $Q_G$, $\mathcal Z_\beta(Q_G)$ encodes the number of independent sets of $G$ (\cref{lem:IS}) or its independence number (\cref{lem:MIS}), yielding the two desired reductions.

The following lemma relates the optimal penalty weight to the probability of sampling a feasible configuration. 
\begin{lemma}[Monotonicity]\label{lem:mono}
Assume $\emptyset\ne\mathcal F\ne \{0,1\}^n$. Then $M\mapsto p_{\beta,M}(\mathcal F)$ is continuous and strictly increasing on $\mathbb R$, and it tends to $0$ as $M\to-\infty$ and to $1$ as $M\to+\infty$. Consequently, for every $\eta\in(0,1)$:
\begin{enumerate}[(i)]
\item $M^*_\eta$ is the unique solution of $p_{\beta,M}(\mathcal F)=\eta$;
\item for every $t\in\mathbb R$: $M^*_\eta\le t \iff p_{\beta,t}(\mathcal F)\ge\eta$;
\end{enumerate}
\end{lemma}
\begin{proof}
Write $\mathcal Z_{\beta,M}=\mathcal Z_\mathcal F+R(M)$ with $R(M) \coloneqq \sum_{x \notin \mathcal F} \e^{ -\beta E^{(o)}(x) -\beta M E^{(p)}(x) }$. Each summand of $R$ is continuous and strictly decreasing in $M$, since $ E^{(p)} \ge 1$. It tends to $0$ as $M\to+\infty$ and to $+\infty$ as $M\to-\infty$. Since $\{0,1\}^n \setminus \mathcal F\neq\emptyset$, $R$ is a continuous strictly decreasing bijection from $\mathbb R$ onto $(0,\infty)$. Since $\mathcal F\ne\emptyset$, $\mathcal Z_\mathcal F>0$, so $p_{\beta,M}(\mathcal F)=\mathcal Z_\mathcal F/(\mathcal Z_\mathcal F+R(M))$ is continuous and strictly increasing with limits $0$ and $1$. By the intermediate value theorem there is a unique $m$ with $p_{\beta,m}(\mathcal F)=\eta$, and $\{M:p_{\beta,M}(\mathcal F)\ge\eta\}=[m,\infty)$, which gives (i) and (ii).
\end{proof}

\begin{definition}[Ancilla-extended instance]\label{def:ancilla}
Let $Q_G\in\mathbb Q^{n\times n}$ with $E_G(\vec x):=\vec x^\top Q_G\vec x$, and let $C,D\in\mathbb Q$. The ancilla extended instance $I(Q_G;C,D)$ of~\eqref{lcBQP} has $n+1$ binary variables $(\vec x,z)\in \{0,1\}^n\times \{0,1\}$, objective
\begin{equation*}
E^{(o)}(\vec x,z):=E_G(\vec x)+C\,z\,|\vec x|-D\,z ,
\end{equation*}
and the single constraint $z=0$. 
\end{definition} 

In other words, the feasible set of $I(Q_G;C,D)$ is $\mathcal F=\{(x,0):x\in \{0,1\}^n\}$. The constraint is defined by $A=(0,\dots,0,1)$ and $\vec b=0$. The objective can be recast as $(\vec x,z)^\top Q\,(\vec x,z)$ with $Q=\begin{psmallmatrix}Q_G & C\mathbf 1\\ 0 & -D\end{psmallmatrix}$. 
The penalty function associated with the constraint of $I(Q_G;C,D)$ is $E^{(p)}(\vec x,z)=z^2=z$.

By construction, we can use the parameters of the ancilla-extended instance $I(Q_G;C,D)$ to control the value of $M^\ast_\eta$. More precisely, we have the following lemma. 
\begin{lemma}[Ancilla gadget]\label{lem:gadget}
Let $Q_G\in\mathbb Q^{n\times n}$, $C,D\in\mathbb Q$, $\beta>0$. Let $\lambda\in\mathbb R$ be such that $E_G(\vec x)+C|\vec x|\ge\lambda$ for all $\vec x\ne\mathbf 0$, and define
$\varepsilon:=\sum_{\vec x\ne\mathbf 0}e^{-\beta(E_G(\vec x)+C|\vec x|)}$. Then,
\begin{equation}
0\le\varepsilon\le(2^n-1)\,e^{-\beta\lambda}\,.
\end{equation}
Moreover, for every $\eta \in (0,1)$, the instance $I(Q_G;C,D)$ satisfies
\begin{equation}\label{eq:closed-form}
M^*_\eta \;=\; D+\frac1\beta\,\ln\!\frac{(1+\varepsilon)\,\eta}{(1-\eta)\,\mathcal Z_\beta(Q_G)}.
\end{equation}
\end{lemma}
\begin{proof}
On $\mathcal F$ we have $E^{(o)}(\vec x,0)=E_G(\vec x)$, hence $\mathcal Z_\mathcal F=\mathcal Z_\beta(Q_G)$. On $\{0,1\}^{n+1} \setminus \mathcal F=\{z=1\}$ we have $E(\vec x,1)=E_G(\vec x)+C|\vec x|-D+M$, hence
\begin{equation*}
\sum_{(\vec x, z) \notin \mathcal F} \e^{ -\beta E(\vec x, z) }=e^{-\beta(M-D)}\sum_{\vec x}e^{-\beta(E_G(\vec x)+C|\vec x|)}=e^{-\beta(M-D)}(1+\varepsilon),
\end{equation*}
where the last equality uses $E_G(\mathbf 0)=0$. 
We have $p_{\beta,M}(\mathcal F)=\frac{\mathcal Z_\beta(Q_G)}{\mathcal Z_\beta(Q_G)+(1+\varepsilon)\,e^{-\beta(M-D)}}$. Since $\emptyset\ne\mathcal F\ne \{0,1\}^{n+1}$, using \cref{lem:mono}, $M^*_\eta$ is the unique solution of $p_{\beta,M}(\mathcal F)=\eta$. This equation reads $e^{-\beta(M-D)}(1+\varepsilon)=\mathcal Z_\beta(Q_G)(1-\eta)/\eta$, which gives ~\eqref{eq:closed-form}.
\end{proof}

For the counting reduction, we choose $Q_G$ such that $\mathcal Z_\beta(Q_G)$ counts the independent sets of $G$ up to an error.
\begin{lemma}[Independent sets as zero-energy states]\label{lem:IS}
Let $G=([n],E)$ be a graph. Let $Q_G$ be its upper-triangular adjacency matrix, so that $E_G(\vec x)=\sum_{\{i,j\}\in E}\vec x_i \vec x_j$. Let $N$ be the number of independent sets of $G$, including $\emptyset$. Then for every $\beta>0$,
\begin{equation*}\label{eq:num_IS_bounds}
N\le\mathcal Z_\beta(Q_G)\le N+2^ne^{-\beta}.
\end{equation*}
\end{lemma}
\begin{proof}
$E_G(\vec x)$ is the number of edges inside $\{i:\vec x_i=1\}$. It is an integer $\ge0$, and it equals $0$ exactly when this set is independent. So the independent sets contribute exactly $N$ to $\mathcal Z_\beta(Q_G) = \sum_\vec x \e^{-\beta E_G(\vec x)}$, and every other term is at most $e^{-\beta}$.
\end{proof}

Using ~\eqref{eq:closed-form} and bounding the approximation errors yields the following theorem. 
\begin{theorem}[$\#\mathsf{P}$-hardness]\label{thm:sharpP}
\textup{\mdseries\textsc{FindM-Dec}}\ is $\#\mathsf{P}$-hard. \textup{\mdseries\textsc{FindM}}\ is $\#\mathsf{P}$-hard even with additive precision $\delta_n:=2^{-n-3}/(n+3)$, where $n+1$ is the number of variables.
\end{theorem}
\begin{proof}
Let $G$ be a graph on $n\ge1$ vertices with $N\in\{1,\dots,2^n\}$ independent sets. Let $Q_G$ be as in \cref{lem:IS}, fix any $D\in\mathbb Q$, and set $\beta=n+3$, $C=2$, and $I=I(Q_G;2,D)$. For $k\in\{0,1,\dots,2^n-1\}$ let 
$\eta_k:=\frac{2k+1}{2k+3}\in(0,1),$ so that $\frac{\eta_k}{1-\eta_k}=k+\tfrac12.$
Since $E_G\ge0$, \cref{lem:gadget} applies with $\lambda=C=2$ and gives $\varepsilon\le 2^ne^{-2(n+3)}\le2^{-n-6}$. \Cref{lem:IS} gives $N\le\mathcal Z_\beta(Q_G)\le N+2^ne^{-(n+3)}\le N+\tfrac18$. By~\eqref{eq:closed-form},
\begin{equation}\label{eq:sharpP-core}
\beta(M^*_{\eta_k}-D)=\ln\frac{(1+\varepsilon)(k+\frac12)}{\mathcal Z_\beta(Q_G)}.
\end{equation}
\emph{Case $N\ge k+1$.} Using $\mathcal Z_\beta(Q_G)\ge k+1$, $k+1\le 2^n$ and $\ln(1+u)\le u ,\; \forall u>-1$,
\begin{equation*}
\beta(M^*_{\eta_k}-D)\le\ln\Big[(1+2^{-n-6})\Big(1-\frac1{2(k+1)}\Big)\Big]\le2^{-n-6}-2^{-n-1}<-2^{-n-2}.
\end{equation*}
\emph{Case $N\le k$.} Using $\mathcal Z_\beta(Q_G)\le k+\frac18$, $\varepsilon \ge 0$ $\ln(1+u)\ge u/(1+u) ,\; \forall u>-1$ and $k+\frac12<2^n$,
\begin{equation*}
\beta(M^*_{\eta_k}-D)\ge\ln\frac{k+\frac12}{k+\frac18}\ge\frac{3/8}{k+\frac12}>2^{-n-2}.
\end{equation*}
Therefore
\begin{equation}\label{eq:sharpP-gap}
M^*_{\eta_k}\le D\iff N\ge k+1,\qquad\text{and}\qquad |M^*_{\eta_k}-D|>2^{-n-2}/\beta=2\delta_n .
\end{equation}
For the decision problem, the query $(I,\beta,\eta_k,t=D)$ to \textup{\mdseries\textsc{FindM-Dec}}
returns whether $N\ge k+1$. Binary search over $k$ determines $N$ with at most $n+1$ adaptive queries.
For the problem of approximation, any valid answer $\tilde M$ to the query $(I, \beta,\eta_k,\delta_n)$ to \textup{\mdseries\textsc{FindM}}  satisfies $|\tilde M-M^*_{\eta_k}|\le\delta_n<|M^*_{\eta_k}-D|$. So $\tilde M\le D\iff N\ge k+1$, and the same binary search works.
All queried instances have size polynomial in $n$ and satisfy the promise. Since computing the number of independent sets of a graph (\textup{\mdseries\textsc{CountIndependentSets}}) is $\#\mathsf{P}$-complete~\cite{Valiant79b,ProvanBall}, every $\#\mathsf{P}$ function is therefore computable in polynomial time with either oracle.
\end{proof}

Note that the reduction uses a low temperature where $\beta$ grows linearly with $n$ and entries of $Q$ are constant in $n$ or, equivalently, $\beta=1$ and entries of $Q$ in $O(n)$.

We further establish \textsf{NP}-hardness of approximation to constant additive error (\textup{\mdseries\textsc{FindM}}) using the same gadget construction but reduction from the decision problem of the independence number. 
To this end, we choose $Q_G$ such that $\mathcal Z_\beta(Q_G)$ pins down the independence number of $G$.
\begin{lemma}[Maximum independent set as ground state]\label{lem:MIS}
Let $G=([n],E)$ have independence number $\alpha=\alpha(G)$, and let $E_G(\vec x):=-|\vec x|+2\sum_{\{i,j\}\in E}\vec x_i\vec x_j$ (the standard QUBO for maximum independent set, cf.~\cite{Lucas_14_Ising}) and $Q_G$ the corresponding matrix. Then $E_G(\vec x)\ge-|\vec x|$ for all $\vec x$, $\min_\vec xE_G(\vec x)=-\alpha$, and hence 
\begin{equation}
e^{\beta\alpha}\le\mathcal Z_\beta(Q_G)\le2^ne^{\beta\alpha}.
\end{equation}
\end{lemma}
\begin{proof}
Let $S=\{i:\vec x_i=1\}$ and let $e(S)$ be the number of edges inside $S$. Then $E_G(\vec x)=-|S|+2e(S)\ge-|S| = -|\vec x|$. Deleting one endpoint of each edge inside $S$ leaves an independent set $S'$ with $|S'|\ge|S|-e(S)$, hence $E_G(\vec x)\ge-|S'|+e(S)\ge-\alpha+e(S)$. So $E_G\ge-\alpha$, with equality exactly when $S$ is a maximum independent set.
The bounds on $\mathcal Z_\beta = \sum_{\vec x} \e^{-\beta E_G(\vec x)}$ follow because the minimizer contributes $e^{\beta\alpha}$ and each of the $2^n$ terms is at most $e^{\beta\alpha}$.
\end{proof}
The following lemma connects the independence number to the optimal penalty weight. 
\begin{lemma}\label{lem:NP}
Fix a rational $\eta\in(0,1)$ and let $\gamma_\eta:=\lceil|\log_2\frac{\eta}{1-\eta}|\rceil$. There is a polynomial-time map $(G,k)\mapsto I_k$ to ancilla-extended instances with $\beta=8(n+\gamma_\eta+1)$ such that
\begin{equation}
\alpha(G)\ge k\ \Longrightarrow\ M^*_\eta(I_k)\le-\tfrac38,\qquad\qquad \alpha(G)\le k-1\ \Longrightarrow\ M^*_\eta(I_k)\ge\tfrac38 .
\end{equation}
\end{lemma}
The map remains polynomial if $\eta$ is part of the input, since $\gamma_\eta$ is at most the bit length of $\eta$.
\begin{proof}
Take $Q_G$ from \cref{lem:MIS} and $I_k:=I(Q_G;\,2,\,k-\frac12)$, so $C=2$ and $D=k-\frac12$. By \cref{lem:MIS}, $E_G(x)+2|x|\ge|x|\ge1$ for $x\ne\mathbf 0$. So \cref{lem:gadget} applies with $\lambda=1$ and gives $\varepsilon\le2^ne^{-\beta}\le1$. Let $c_\eta:=\ln\frac{\eta}{1-\eta}$, so $|c_\eta|\le\gamma_\eta$. By~\eqref{eq:closed-form} and $\beta\alpha\le\ln\mathcal Z_\beta(Q_G)\le\beta\alpha+n\ln2$ (\cref{lem:MIS}),
\begin{align*}
M^*_\eta&=D+\tfrac1\beta\big(c_\eta+\ln(1+\varepsilon)-\ln\mathcal Z_\beta(Q_G)\big)\ \le\ D-\alpha+\tfrac{\gamma_\eta+1}{\beta}\ \le\ D-\alpha+\tfrac18,\\
M^*_\eta&\ge D-\alpha-\tfrac{\gamma_\eta+n\ln2}{\beta}\ \ge\ D-\alpha-\tfrac18 .
\end{align*}
Taking $D=k-\frac12$, we obtain if $\alpha\ge k$ then $M^*_\eta\le-\frac12+\frac18=-\frac38$, and if $\alpha\le k-1$ then $M^*_\eta\ge\frac12-\frac18=\frac38$.
\end{proof}
In this way, the approximation problem of optimal penalty weight inherits the hardness of the related problem for the independence number and we formally establish the following theorem for approximation to within a constant additive error. 
\begin{theorem}[$\mathsf{NP}$-hardness of approximation]\label{thm:NP}
For every fixed rational $\eta\in(0,1)$,
computing $M^*_\eta$ to within additive error $\frac14$ (indeed any error $<\frac38$) is $\mathsf{NP}$-hard.
\end{theorem}
The constant $\frac38$ reflects the normalization of the instances of \cref{lem:NP} (couplings $-1$ and $2$); as discussed later, approximating $M^*_\eta$ within any other additive error is equally hard.
\begin{proof}
By the proof of \cref{lem:NP}, $|M^*_\eta-(D-\alpha)|\le\frac18$. So using a \textup{\mdseries\textsc{FindM}} oracle with $\delta < \frac{3}{8}$ gives $|\tilde M-M^*_\eta|<\frac38$, which in turn gives $|(D-\tilde M)-\alpha|<\frac12$. 
Thus, $\alpha$ is the integer nearest to $D-\tilde M$, so the oracle output determines $\alpha(G)$, answering the problem \textup{\mdseries\textsc{MaxIndependentSet}}, which is $\mathsf{NP}$-hard.
\end{proof}

Finally, we discuss the units. The Gibbs distribution depends on the energies only through $\beta E$. Hence, for every $s>0$, the rescaling $(Q,M,\beta)\mapsto(sQ,sM,\beta/s)$ leaves $p_{\beta,M}(\mathcal F)$ unchanged, and $M^*_\eta(sQ,\beta/s)=s\,M^*_\eta(Q,\beta)$. Additive errors on $M^*_\eta$ therefore carry the units of the couplings. Rescaling the instances of \cref{lem:NP} by a rational $s$ turns the gap $\pm\frac38$ into $\pm\frac{3s}8$, so approximating $M^*_\eta$ within any additive error $\delta$ is $\mathsf{NP}$-hard.
The rescaling leaves $\beta M^*_\eta$ invariant, so a scale-free measure of an error $\delta$ on $M^*_\eta$ is $\beta\delta$, i.e.\ $\delta$ in units of the temperature $1/\beta$. In these units, \cref{thm:sharpP} requires precision $\beta\delta_n=2^{-n-3}$, whereas \cref{thm:NP} gives $\mathsf{NP}$-hardness already for an error strictly below $\frac38\beta=3(n+\gamma_\eta+1)$.
}

\section{Extended theoretical guarantee for the algorithm} \label{app:guarantee_realistic}
{
Theorem~\ref{thm:guarantee} establishes that the proposed algorithm returns a value $M^\ast$ that ensures an $\eta$-reformulation, \cref{eq:eta_reformulation}, using infinite samples $N_s \to \infty$. 
We here establish that this guarantee extends to any finite number of samples if the margins $\epsl$ and $\epsh$ of \cref{alg:M} are chosen appropriately. 
The margins turn the estimate of each feasible spectral weight into a lower confidence bound where it enters $\Bfl$, and into an upper confidence bound where it enters $\Bfh$. 
We propose a choice of the margins based on the Clopper-Pearson bounds, show that it yields an $\eta$-reformulation with high probability for every $N_s$ and $\Delta$, and quantify by how much the finite sample size increases the returned penalty. 

Let $S = \{\vec x_1,\dots, \vec x_{N_s}\}$ be a sequence of $N_s$ samples drawn independently and uniformly from the feasible set $\mathcal F$. 
We define the empirical estimator $\hat n_\Delta(e) = |\{\vec x \in S : e \le E^{(o)}(\vec x) < e + \Delta\}| \frac{|\mathcal{F}|}{N_s}$, where $|\mathcal{F}| = n_\mathrm{pen}(0)$ denotes the number of feasible points. 
For structured problems, $|\mathcal{F}|$ can often be computed analytically (see \cref{app:penalization_degeneracies}). 
Throughout this section, the resolution $\Delta$ and the threshold $E_f$ are chosen independently of the samples, so that the bins $b(e)$ with $e \in \{0, \Delta, 2\Delta, \dots\}$ and their assignment to $\LambdaLow$ are fixed. 
For such a bin, we denote by $q_e \coloneqq n_\Delta(e+\ELB)/|\mathcal F|$ the probability that a uniform sample from $\mathcal F$ falls into it and by $X_e \coloneqq N_s\, \hat n_\Delta(e + \ELB)/|\mathcal F|$ the number of samples that do, such that $X_e \sim \mathrm{Bin}(N_s, q_e)$ is binomially distributed. 

For $\delta' \in (0,1/2]$, $N \in \mathbb N$ and $k \in \{0,\dots,N\}$, the Clopper-Pearson bounds~\cite{Clopper1934} are defined as follows: 
$\underline q_{\delta'}(0) \coloneqq 0$ and, for $k\geq 1$, $\underline q_{\delta'}(k)$ is the unique $q\in(0,1)$ with $\Pr_{X \sim \mathrm{Bin}(N,q)}[X \geq k] = \delta'$; $\overline q_{\delta'}(N) \coloneqq 1$ and, for $k < N$, $\overline q_{\delta'}(k)$ is the unique $q\in(0,1)$ with $\Pr_{X \sim \mathrm{Bin}(N,q)}[X \leq k] = \delta'$. 
Equivalently, $\underline q_{\delta'}(k)$ is the $\delta'$-quantile of the $\mathrm{Beta}(k, N-k+1)$ distribution and $\overline q_{\delta'}(k)$ the $(1-\delta')$-quantile of the $\mathrm{Beta}(k+1, N-k)$ distribution and can, thus, be evaluated by  standard numerical libraries. 
The following lemma shows that the Clopper–Pearson bounds lie below and above the true parameter, respectively, each with probability at least $1-\delta'$, and gives explicit bounds on their distances.

\begin{lemma}[Clopper-Pearson bounds]\label{lem:cp}
Let $X \sim \mathrm{Bin}(N,q)$, $\hat q \coloneqq X/N$ and $L \coloneqq \ln(1/\delta')$. Then 
(a) $\Pr[\underline q_{\delta'}(X) \leq q] \geq 1-\delta'$ and $\Pr[q \leq \overline q_{\delta'}(X)] \geq 1-\delta'$, and 
(b) for every outcome of $X$, 
\begin{equation}\label{eq:cp_width}
    \hat q - \sqrt{\frac{2 \hat q L}{N}} \;\leq\; \underline q_{\delta'}(X) \;\leq\; \hat q \;\leq\; \overline q_{\delta'}(X) \;\leq\; \hat q + \sqrt{\frac{2\hat q L}{N}} + \frac{2L}{N}\,.
\end{equation}
\end{lemma}
\begin{proof}
(a) For $k \geq 1$, the function $f_k(q) \coloneqq \Pr_{X \sim \mathrm{Bin}(N,q)}[X \geq k]$ is continuous and strictly increasing in $q$, and non-increasing in $k$. 
Hence, $\underline q_{\delta'}(k) > q$ if and only if $f_k(q) < \delta'$, and the set of such $k$ is of the form $\{k \geq k_0\}$. 
Therefore, $\Pr[\underline q_{\delta'}(X) > q] = \Pr[X \geq k_0] = f_{k_0}(q) < \delta'$. 
The statement for $\overline q_{\delta'}$ follows by applying this to $N - X \sim \mathrm{Bin}(N, 1-q)$. 
(b) Since $f_k(k/N) \geq 1/2 \geq \delta'$, we have $\underline q_{\delta'}(k) \leq \hat q$, and analogously $\overline q_{\delta'}(k) \geq \hat q$. 
Let $D(a\|q)$ denote the binary relative entropy. 
By the Chernoff bound, $f_k(q) \leq \exp[-N D(\hat q\|q)]$ for $q \leq \hat q = k/N$, and $D(a\|q) \geq (a-q)^2/(2a)$ for $q \leq a$. 
Evaluating at $q = \underline q_{\delta'}(k)$, where $f_k = \delta'$, yields $(\hat q - \underline q_{\delta'}(k))^2 \leq 2\hat q L/N$. 
Likewise, $\Pr_{X \sim \mathrm{Bin}(N,q)}[X \leq k] \leq \exp[-N D(\hat q\|q)]$ for $q \geq \hat q$ and $D(a\|q) \geq (q-a)^2/(2q)$ for $q \geq a$ yield $(\overline q_{\delta'}(k) - \hat q)^2 \leq 2\, \overline q_{\delta'}(k) L/N$, which implies the last inequality. 
\end{proof}

To account for feasible points with objective values beyond the sampled range, the following lemma gives a high-probability bound on their fraction.
\begin{lemma}[Sample maximum tail bound]\label{lem:tail}
Let $E_{\max}$ be the largest objective value among the $N_s$
samples, and let $\delta'\in(0,1)$. With probability at least
$1-\delta'$,
\[
\Pr_{\mathbf{x}\sim\operatorname{Unif}(\mathcal{F})}
\bigl[E^{(o)}(\mathbf{x})>E_{\max}\bigr]
\le 1-\delta'^{1/N_s}
\le \frac{\ln(1/\delta')}{N_s}.
\]
\end{lemma}

\begin{proof}
Set $t:=1-\delta'^{1/N_s}$, and let $e_t$ be the smallest
attained objective value such that
$\Pr_{\mathbf{x}\sim\operatorname{Unif}(\mathcal{F})}
[E^{(o)}(\mathbf{x})>e_t]\le t$.
The first inequality can fail only if $E_{\max}<e_t$,
which, by the minimality of $e_t$, occurs with probability
at most $(1-t)^{N_s}=\delta'$.
The second inequality follows from $1-e^{-u}\le u$.
\end{proof}

We can now choose the margins in the input of \cref{alg:M} such that we replace the estimates in lines~\ref{alg:line:Bf_<} and \ref{alg:line:Bf_>} of \cref{alg:M} by the (rescaled) lower and upper Clopper-Pearson bounds with $N = N_s$, respectively. Given a confidence parameter $\delta \in (0,1)$, we assign to the $l$-th bin, $e = l\Delta$ with $k = 0, 1, 2, \dots$, the parameter $\delta_l \coloneqq \delta/[2(l+1)(l+2)]$, such that $\sum_{l \geq 0} \delta_l = \delta/2$, and set 
\begin{equation}\label{eq:margins}
\begin{split}
    \epsl(e) &\coloneqq \hat n_\Delta(e + \ELB) - |\mathcal F|\, \underline q_{\delta_l}(X_e)\,, \\
    \epsh(e) &\coloneqq |\mathcal F|\, \overline q_{\delta_l}(X_e) - \hat n_\Delta(e + \ELB) \,,
\end{split}
\end{equation}
which are non-negative by \cref{eq:cp_width}. 
For finite $E_f$, we additionally have to account for feasible points above the sampled range. 
To this end, we append to $\LambdaHigh$ the point $e_\infty \coloneqq \max \Lambda + \Delta$ with the bin $b(e_\infty) \coloneqq \{\vec x \in \mathcal F: E^{(o)}(\vec x) - \ELB \geq e_\infty\}$. 
This bin contains no samples and we set 
\begin{equation}\label{eq:margin_tail}
    \epsh(e_\infty) \coloneqq |\mathcal F| \bigl(1 - (\delta/2)^{1/N_s}\bigr) \;\leq\; |\mathcal F|\, \frac{\ln(2/\delta)}{N_s}\,.
\end{equation}
For vanishing margins the appended point does not contribute, so that \cref{alg:M} with $\epsilon = 0$ is not affected by this extension. 
With these choices of margins, we can establish our extended guarantee for \cref{alg:M}: 

\begin{theorem}[Finite-sampling guarantee] \label{thm:finite_sampling}
    Let $N_s \in \mathbb N$, $\delta \in (0,1)$ and $\vcut = \max_{\vec x} E^{(p)}(\vec x)$, and let $\Delta > 0$ and $E_f$ be chosen independently of the samples. 
    With probability at least $1-\delta$ over the samples the following holds: 
    If \cref{alg:M} with the margins \eqref{eq:margins} and \eqref{eq:margin_tail} returns $M^\ast \neq \{\}$, then \eqref{QUBO} with $M = M^\ast$ is an $\eta$-reformulation with guaranteed energy threshold $E_f$ for a Gibbs sampler at inverse temperature $\beta$. 
\end{theorem}
\begin{proof}
    Whether a bin of the lattice $\{0, \Delta, 2\Delta, \dots\}$ lies entirely below $E_f$ does not depend on the samples. 
    By \cref{lem:cp}(a) and a union bound over all bins, with probability at least $1 - \sum_{l \geq 0} \delta_l = 1 - \delta/2$, we have 
    \begin{equation}\label{eq:good_event}
    \begin{split}
        |\mathcal F|\, \underline q_{\delta_l}(X_e) &\leq n_\Delta(e+\ELB) \quad \text{for all bins entirely below $E_f$,}\\
        |\mathcal F|\, \overline q_{\delta_l}(X_e) &\geq n_\Delta(e+\ELB) \quad \text{for all other bins.}
    \end{split}
    \end{equation}
    Moreover, all points in $b(e_\infty)$ have an objective value above the largest sampled one, so that $|b(e_\infty)| \leq \epsh(e_\infty)$ with probability at least $1-\delta/2$ by \cref{lem:tail} with $\delta' = \delta/2$. 
    Both statements hold simultaneously with probability at least $1-\delta$, which we assume in the following. 

    Let $\Bfl$ and $\Bfh$ denote the quantities computed in lines~\ref{alg:line:Bf_<} and \ref{alg:line:Bf_>} of the algorithm. 
    By \cref{eq:margins,eq:good_event} and the argument leading to \cref{eq:Bfl}, 
    \begin{equation}\label{eq:Bfl_hat}
        c\, \Bfl \;\leq\; c \sum_{e \in \LambdaLow} \e^{-\beta(e+\Delta)}\, n_\Delta(e + \ELB) \;\leq\; p\,.
    \end{equation}
    Furthermore, every feasible point with an objective value larger than $E_f$ lies in a bin $b(e)$ with $e \in \LambdaHigh \cup \{e_\infty\}$. 
    Hence, as in \cref{eq:Bfh}, 
    \begin{equation}\label{eq:Bfh_hat}
        \Pr[\efh] \;\leq\; c\, \Bigl( \sum_{e \in \LambdaHigh} \e^{-\beta e}\, n_\Delta(e+\ELB) + \e^{-\beta e_\infty} |b(e_\infty)| \Bigr) \;\leq\; c\, \Bfh\,.
    \end{equation}
    For $E_f = \infty$ both sides vanish. 
    Since \cref{eq:Binf} holds for every $M$, and thus also for the sample-dependent output $M^\ast$, the remaining steps of the proof of Theorem~\ref{thm:guarantee}, \cref{eq:Bfh+Binf,eq:compound_inequality_proof}, apply verbatim for the current finite-sample $\Bfl$ and $\Bfh$, and yield $p \geq \eta$. 
\end{proof}

Note that Theorem~\ref{thm:finite_sampling} requires no assumption on $N_s$, $\Delta$ or $\beta$, yields the target $\eta$ itself, and depends on $\delta$ only logarithmically. 
The sample size only determines how conservative the returned penalty is. 
By \cref{eq:cp_width}, the margins \eqref{eq:margins} of a bin that contains $X_e \geq 1$ samples amount to at most a fraction $\sqrt{2L_l/X_e} + 2L_l/X_e$ of its estimate, where $L_l \coloneqq \ln(1/\delta_l)$, while an empty bin has $\epsl = 0$ and $\epsh \leq |\mathcal F| L_l/N_s$. 
Since $2(l+1)(l+2)\le 4|\Lambda|^2$ for every bin of $\Lambda$, $L_l\le 2\ln(|\Lambda|/\delta)+\ln 4$, so these relative margins are small for bins that contain a fraction of the feasible points much larger than $\ln(|\Lambda|/\delta)/N_s$.
The margins thus vanish for $N_s \to \infty$, where Theorem~\ref{thm:guarantee} is recovered. 
Since the outer expressions in \cref{eq:cp_width} enclose the Clopper-Pearson bounds, Theorem~\ref{thm:finite_sampling} also holds if the latter are replaced by the explicit values $\max\{0, \hat q_e - \sqrt{2 \hat q_e L_l/N_s}\}$ and $\hat q_e + \sqrt{2 \hat q_e L_l/N_s} + 2L_l/N_s$ with $\hat q_e \coloneqq X_e/N_s$, at the price of larger margins. 
Moreover, their effect on the output is only logarithmic, as we show next. 
Let $v_1 \coloneqq \min\{v \geq 1: \nv(v) > 0\}$. 

\def\Bflinf{\ensuremath  B^{\scalebox{.7}{<}}_{\infty}}
\def\Bfhinf{\ensuremath  B^{\scalebox{.7}{>}}_{\infty}}
\begin{proposition}[Sampling-induced penalty error]\label{prop:finite_sampling_cost}
    Let $\Bflinf$ and $\Bfhinf$ be the quantities of \cref{alg:M} for exact spectral weights, as in Theorem~\ref{thm:guarantee}, and $\Bfl$ and $\Bfh$ the ones computed from $N_s$ samples with the margins \eqref{eq:margins} and \eqref{eq:margin_tail}. 
    Set $G_\infty \coloneqq \frac{1-\eta}{\eta}\Bflinf - \Bfhinf$ and $G \coloneqq \frac{1-\eta}{\eta}\Bfl - \Bfh$, and let $M^\ast_\infty$ and $M^\ast$ be the corresponding roots of $\Binf(M) = G_\infty$ and $\Binf(M) = G$, assumed to exist and to be positive. 
    Then, with probability at least $1-\delta$, 
    \begin{equation}\label{eq:finite_sampling_cost}
        M^\ast_\infty \;\leq\; M^\ast \;\leq\; M^\ast_\infty + \frac{1}{\beta v_1} \ln \frac{G_\infty}{G}\,.
    \end{equation}
\end{proposition}
\begin{proof}
    On the event considered in the proof of Theorem~\ref{thm:finite_sampling}, we have $\Bfl \leq \Bflinf$ and $\Bfh \geq \Bfhinf$, and hence $0 < G \leq G_\infty$. 
    Since $\Binf$ is decreasing, this gives the first inequality. 
    Every term of $\Binf(M) = \sum_{v \geq v_1} \e^{-\beta M v} \nv(v)$ decays at least as $\e^{-\beta a v_1}$ under $M \mapsto M + a$ with $a \geq 0$, such that $\Binf(M^\ast_\infty + a) \leq \e^{-\beta a v_1} G_\infty$. 
    For $a = \ln(G_\infty/G)/(\beta v_1)$ the right-hand side equals $G = \Binf(M^\ast)$, and thus $M^\ast \leq M^\ast_\infty + a$. 
\end{proof}

By \cref{eq:finite_sampling_cost}, the statistical error enters the returned penalty only logarithmically, and $G$ approaches $G_\infty$ once the bins that dominate $\Bfl$ and $\Bfh$ contain many more than $L_k$ samples each. 
Since Theorem~\ref{thm:finite_sampling} holds for every $\Delta$, the resolution can be chosen to balance two effects. 
A coarse lattice underestimates each Boltzmann weight entering $\Bfl$ by up to a factor $\e^{-\beta\Delta}$, which by the argument above raises the penalty by at most $\Delta/v_1$ for $E_f = \infty$. 
A fine lattice leaves fewer samples per bin and, thus, leads to larger margins. 
With the margins \eqref{eq:margins}, a resolution of the order of the thermal energy, $\Delta \approx \beta^{-1}$, is therefore a natural choice. 
For vanishing margins, in contrast, $\Bfl$ and $\Bfh$ are plain Monte-Carlo estimates and a fine lattice comes at no cost. 
This is the heuristic setting of our numerical experiments, where we set $\epsilon = 0$: it is covered by Theorem~\ref{thm:guarantee} in the limit $N_s \to \infty$.
}

\section{Benchmarking problems} \label{app:benchmark_problems}

In the present section, we discuss the benchmarked optimization problems and their QUBO formulations.

\subsection{MNPP}
The Number Partitioning Problem (NPP), an NP-hard combinatorial problem, aims at partitioning a set of numbers into two subsets as evenly as possible. The Multiway Number Partitioning Problem (MNPP) is a generalization of this, with multiple subsets to partition the elements into. 
Its applications range from distributed networking and computing, resource allocation and logistics, to gerrymandering and investment portfolio \cite{Schreiber_18_MNPP}.

More formally, given a set $S$ of $N$ positive numbers $S = \{ c_{1}, \ldots , c_{N}\}$, the goal is to partition of $S$ into $P$ disjoint subsets $R_1, \ldots , R_{P}$, such that the sums of values in each subset are as close to each other as possible. This problem can be stated as follows: can a set of $N$ assets with values $c_1, \ldots , c_{N}$ fairly be distributed between $P$ parties?
To model the problem in an optimization context \cite{Lucas_14_Ising}, we define the $NP$ binary decision variables $\{x_{i, p}\}_{i\in \{1,\dots,N\} , p \in \{1, \dots, P\} } \in \{0,1\}^{NP}$, assigning each element $c_i$ in $S$ to a subset $R_p$, defined so that
\begin{align}
	x_{i, p} = 
    \begin{cases}
	1 & \text{if  } c_i \in R_p \\
	0 & \text{otherwise}
    \end{cases} \, .
\end{align}
The constraints will encode the fact that each element can only be assigned to one subset, that is, the decision matrix $[x]_{i, p}$ is a right-stochastic matrix: $\sum_{p=1}^{P} x_{i, p} = 1 \; \forall i$. The objective function can be stated in different ways \cite{Korf_10_MNPP}; we will consider
\begin{equation}
    E^{(o)}(x) = \sum_{p=1}^{P}  \left( \sum_{i=1}^{N} c_{i} x_{i, p} - \frac{1}{P} \sum_{i=1}^{N} c_{i} \right)^2
\end{equation}
that sums the squared errors of the subsets sums with respect to a perfectly even distribution of $ \frac{1}{P} \sum_{i=1}^{N} c_{i}$ per subset (scaled variance of the subset sums).

The QUBO model can be formulated in the following way:
\begin{align}
	\min_{x \in \{0,1\}^{NP}} E^{(o)}(x) + M \sum_{i=1}^{N}  \left( 1- \sum_{p=1}^{P} x_{i, p} \right)^2 \, ,
\end{align}
In the benchmarks considered in this work, the numbers $c_i$ to be partitioned were randomly generated from a uniform distribution over the interval $[0, 10^3]$. For small-scale tests involving the Gibbs sampler, the number of partitions was fixed to $P = 3$, while only the system size $N$ was increased. Conversely, for the larger-scale tests with the SA and DA solvers, both $N$ and $P$ were increased with system size, maintaining the relation $N = 8P$. This choice ensured that the average number of integers per partition remained constant as the system grew, preventing the problem from becoming artificially easier and avoiding an overabundance of optimal solutions \cite{Gent_NPPtrans_96}.

\subsection{TSP}
The Traveling Salesman Problem is a cornerstone of optimization problems, it is a NP-hard problem highly relevant both in theoretical computer science and for practical applications, such as logistics, circuit design, telecommunications \cite{Matai_10_TSPappl}. Given a connected graph $G=(V,E)$ with $n_v =|V|$ vertices, where edge $e_{i,j}$ represents the cost of traveling from node $i$ to node $j$, the goal is to find the cheapest Hamiltonian cycle, i.e. the path that visits all the nodes in the graph, minimizing the overall total travel cost. The combinatorial problem can be encoded via the $n_v^2$ binary decision variables $\{x_{t, i}\}_{t, i=1,\dots,n_v} \in \{0,1\}^{n_v^2}$, ordering the temporal visit of each city \cite{Lucas_14_Ising}, defined so that
\begin{align}
	x_{t, i} = 
    \begin{cases}
	1 & \text{if city } i \text{ is visited at time step } t \\
	0 & \text{otherwise}.
    \end{cases}
\end{align}
The constraints of the optimization problem enforce the decision matrix $[x]_{t, i}$ to be a permutation matrix,  that is $\sum_{i=1}^{n_v} x_{t, i} = 1 \; \forall t$ and $\sum_{t=1}^{n_v} x_{t, i} = 1 \; \forall i$. The objective function to minimize incorporates the cost (given by the sum of the edge weights) of a path represented by a particular realization of the decision matrix \cite{Lucas_14_Ising}
\begin{equation}
    E^{(o)}(x) = \sum_{t=1}^{n_v}   \sum_{i\ne j=1}^{n_v} e_{i,j} x_{t, i}x_{t+1, j}.
\end{equation}
The QUBO model is then formulated in the following way:
\begin{align}
	\min_{x \in \{0,1\}^{n_v^2}} E^{(o)}(x)  + M \bigg[ \sum_{i=1}^{n_v}  ( 1- \sum_{t=1}^{n_v} x_{t, i} )^2 + \sum_{t=1}^{n_v}  ( 1- \sum_{i=1}^{n_v} x_{t, i} )^2  \bigg]\, .
\end{align}
In the present work, multiple sets of TSP benchmarks are used or generated. The first, referred to as the \emph{circle TSP} in the plots, consists of $n_v$ nodes deterministically positioned at equal distances along a circle of radius $10^6$. The second set, called the \emph{random TSP}, contains instances where nodes are randomly placed within a square of side length $2 \times 10^6$, while the third set, named the \emph{benchmark TSP}, includes instances obtained from the standard benchmark library \cite{TSPbenchmarks}.

\subsection{PO}
For Portfolio Optimization we use the well-known Markovitz model \cite{Markowitz_52_PO, Grant_21_markovitz, Rosenberg_16_risk}, i.e.\ the problem of selecting a set of assets maximizing returns while minimizing risk.
The problem specification requires a vector $\boldsymbol \mu$ of expected returns of a set of $N$ assets, their covariance matrix $\Sigma$, a risk aversion $\gamma > 0$, and a partition number $w$ defining the portfolio discretization.  
Denoting by $x_i$ the units of asset $i$ in the portfolio, the problem formulation reads
\begin{equation}
    \operatorname*{minimize}_{\vec x \in \mathbb N^N}\ -\boldsymbol{\mu}^t \vec x + \gamma\, \mathbf{x}^T \Sigma \mathbf{x} \quad 
\operatorname{subject\,to}\quad  \textstyle\sum_{i=1}^N x_i = 2^w - 1\,. \label{Markowitz_problem}
\end{equation}
The constraint forces the budget to be totally invested. 
The QUBO reduction requires mapping each integer decision variable into $w$ binary variables. 
We generate problem instances from historic financial data on S\&P 500 stocks.

In what follows, we illustrate how the data used in Portfolio Optimization instances were fetched from real data and adapted to Markowitz formulation \eqref{Markowitz_problem}. From stock market index S\&P500, we downloaded the stock price history, referring to the 2 years period December 2020 until November 2022 with one-month interval, of 121 out of the 500 company stocks tracked by S\&P500 (namely, the ones with no missing data in said intervals). Let us call $P_{t,a}$ such cost of an asset $a$, with time index $t$. The return at time step $t$ is defined as
\begin{equation}
    r_{t,a} = \frac{P_{t,a} - P_{t-1,a}}{P_{t-1,a}}
\end{equation}
from which the expected return vector $\tilde{\boldsymbol{\mu}}$ and the covariance matrix $\tilde{\Sigma}$ can be computed as $
    \tilde{\mu}_a = \frac{1}{T}\sum_{t = 1}^T r_{t,a}$ 
and $
    \tilde{\Sigma}_{a,b} = \frac{1}{T-1}\sum_{t = 1}^T (r_{t,a} - \mu_a)(r_{t,b} - \mu_b)$.
We encode the real financial stock market data with decimal precision of $10^{-4}$. 

Another parameter of the generated instances is the partition number $w$ \cite{Grant_21_markovitz}, that describes the granularity of the portfolio discretization, since the budget is divided in $2^w - 1$ equally large chunks. Each asset decision variable $x_i$ is an integer that can take values from $0$ up to $2^w - 1$, indicating how many of these partitions to allocate towards asset $i$. This explains why the constraint $\sum_i x_i = 2^w - 1$ enforces the budget to be totally invested.
As a consequence, $w$ is also equal to the number of bits one needs to allocate for every integer and, by extension, asset.

Notice that $\tilde{\boldsymbol{\mu}}^t \boldsymbol{p}$ is the expected return of a portfolio if $\boldsymbol{p}$ represents the vector of the \emph{portions} of the portfolio for each asset, i.e.\ $0 \le p_i \le 1$ and $\sum_i p_i = 1$. In order to have integer decision variables, the number of chunks $x_i = (2^w-1)p_i$ is used, and in the final formulation (\ref{Markowitz_problem}) of the Markowitz model the factors are absorbed in the objective function, defining $\boldsymbol{\mu} = \tilde{\boldsymbol{\mu}}/(2^w-1)$ and $\Sigma = \tilde{\Sigma}/(2^w-1)^2$.

The last parameter that one needs to set to fully specify the instance is the risk aversion factor $\gamma$, weighting differently the return and the volatility in the objective function. Common values of the risk aversion factor are $\gamma = 0.5,1,2$.

The parameter values used in the instances tested in this work are $\gamma = 1$ and $w = 3$ $(5)$ for the small- (large-) scale tests employing the Gibbs (SA and DA) solvers, respectively, yielding a portfolio granularity of $7$ ($31$) equal chunks.

\section{Details on the algorithm subroutines} \label{app:subroutines}
This section elaborates on the subroutines that constitute \cref{alg:M}.

In line~\ref{alg:line:pen_deg}, we take as input the penalization function $E^{(p)}$ of the problem under consideration and the violation threshold $v_\mathrm{cut}$, which sets the maximum penalization value considered. For the constrained problems analyzed in this work, the penalization degeneracies values $n_\mathrm{pen}(v)$ were analytically derived in \cref{app:penalization_degeneracies} from the structure of $E^{(p)}$. Using these expressions, one can directly construct the vector of degeneracies up to $v_\mathrm{cut}$. For more general problems not considered here, where analytical degeneracies are unavailable, one can instead estimate the penalization degeneracies by evaluating the penalization energies of uniformly sampled bitstrings followed by a fitting procedure.

Line~\ref{alg:line:ELB} computes a lower bound of the objective function $E_\mathrm{LB} \le \min_{\vec x \in \{0,1\}^n} E^{(o)}(\vec x)$. In principle, any valid method to compute such a lower bound can be employed; the tighter the bound, the better the algorithm's performance. In this work, we use a Semi-definite Programming (SDP) relaxation for the \ac{PO} case, where an analytical lower bound is not evident, while we set $E_\mathrm{LB} = 0$ for the other cases, \ac{TSP} and \ac{MNPP}, where this value is an obvious bound (in particular, for \ac{TSP}, where $\min_{\vec x \in \{0,1\}^n} E^{(o)}(\vec x) = 0$). The SDP relaxation consists of optimizing over the cone of semi-definite matrices intersected with linear constraints. Specifically, to lower bound $E^{(o)}(\vec x) = \vec x^t Q \vec x + L^t \vec x$, an SDP relaxation is formulated as
\begin{align}
    E_\mathrm{LB} \coloneqq \min_Y\ &\Tr\{Y^t\tilde{Q}\} \\
    \text{s.t.} \;  &Y \ge 0\\
    &Y_{1i} =Y_{ii} \quad \forall i=2,\dots,n+1 \\
    &  Y_{11} = 1,
\end{align}
where $Y$ is a $(n+1)\times(n+1)$ real positive semidefinite matrix, $ Tr\{A^tB\} = \langle A,B\rangle = \sum_{ij} A_{ij}B_{ij}$ denotes the inner product between matrices and the matrix $\tilde{Q}$ is given by
\begin{equation}
\tilde{Q} = \begin{bmatrix} 
0 & \frac{1}{2}L^t \\
\frac{1}{2}L & Q
\end{bmatrix}.
\end{equation}
For further details on the derivation, see Ref.~\cite{BigM_25}, Appendix B.

In line~\ref{alg:line:Usample}, we compute a pre-defined number $N_s$ of objective energy values $\{E_i\}_{i=1}^{N_s}$ of feasible bitstrings uniformly drawn from $\mathcal F$ and approximate the feasible spectral weights \cref{eq:feasible_spectrum}. For a general constrained problem with constraints of the form $A \vec x = b$, uniformly sampling feasible bitstrings is generally computationally challenging \cite{Jerrum_UnifSampl_86, Jerrum_UnifSampl_03}. Nevertheless, for structured problems, efficient strategies can often be devised \cite{Vigoda_countKnap_14, Pesant_countCSP_15}. In the present work, such strategies were implemented depending on the specific problem structure.
For \ac{MNPP}, each feasible configuration corresponds to an assignment of all $N$ items into $P$ partitions, yielding $P^N$ possible solutions. Uniform sampling is achieved by independently assigning each item $i$ to a random partition $p(i) \in {1,\dots,P}$ and setting $\vec x_{i,p(i)} = 1$ (and $0$ otherwise).
For \ac{TSP}, a feasible bitstring represents a Hamiltonian cycle, which can be uniformly generated by sampling a random permutation $\sigma$ of ${1,\dots,n_v}$ and setting $\vec x_{\sigma(i), i} = 1$ and $0$ otherwise.
For \ac{PO}, each valid portfolio corresponds to one of the $\binom{2^w + N - 2}{N - 1}$ ways of distributing $2^w-1$ indistinguishable units among $N$ assets.
After generating $N_s$ uniform samples, associated objective energies are computed. From this, the computation of $n_\Delta(e)$ follows from the definition \cref{eq:feasible_spectrum}.

Steps~\ref{alg:line:Binf}, \ref{alg:line:Bf_<} and \ref{alg:line:Bf_>} use the computed $n_\Delta(e)$ and $n_\mathrm{pen}(v)$ to evaluate, up to a common factor, the probability bounds for three classes of configurations: infeasible states, feasible states with low objective energy, and feasible states with high objective energy.

Finally, in step~\ref{alg:line:g_root}, the numerical root-finding of $g$ is straightforward, enabled by the function’s favorable analytical properties.

\section{Direct estimation of  penalization weights}\label{app:M_l1}

In order to establish a baseline for our algorithm, we here develop a strategy for determining $M$ based on standard, simple bounds of the objective function to the case of Gibbs samplers.  
This strategy captures typical initial considerations of practitioneers when determining values of $M$ using binary search. 
We provide a simple efficient formula for a penalty weight $M$ that yields an $\eta$-reformulations for an exact Gibbs solver at known temperature.  
The penalty weight is the sum of two terms: First, the penalty weight that is widely used \citep{Harwood_routing, Leonidas_vehiclerouting, qiskit_doc, Garcia_penalty_22} and sufficient for an exact solver, i.e. a Gibbs solver at inverse temperature $\beta = \infty$. Second, a thermal correction proportional with the temperature of the solver:
\begin{equation} \label{eq:M_l1_app}
    M_{\ell_1}(\beta) = \beta^{-1}(n \ln 2 -\ln (1-\eta)) + \norm{Q}_{\ell_1}
\end{equation}
The expression is motivated by the following guarantee:
\begin{lemma}\label{lem:Ml1}
 For a constrained problem \eqref{lcBQP} of size $n$ and objective function $E^{(o)}(\vec x) = \vec x^t Q \vec x$ a penalty weight $M_{\ell_1}(\beta)$ defined in \eqref{eq:M_l1_app} ensures an $\eta$-formulation \eqref{QUBO} for a Gibbs solver at inverse temperature $\beta$.
\end{lemma}

\begin{proof}
    Ensuring that the probability of sampling feasible solutions is greater than $\eta$ is equivalent to show that sampling infeasible solutions occurs with probability at most $1-\eta$, i.e. $\mathbb P[\vec x \notin \mathcal F] \le 1-\eta$. Using the fact that the energy of the \ac{QUBO} formulation will be $E(\vec x) = E^{(o)}(\vec x) + M(A \vec x - b)^2$ and for infeasible points $(A \vec x - b)^2 \ge 1$,  we can bound:
    \begin{equation}
        \mathbb P[\vec x \notin \mathcal F] = \sum_{\vec x \notin \mathcal F} p(\vec x) = \mathcal N_\beta \sum_{\vec x \notin \mathcal F} e^{-\beta  \bigl(E^{(o)}(\vec x) + M(A \vec x - b)^2\bigr)} \le \mathcal N_\beta e^{-\beta M}\sum_{\vec x \notin \mathcal F} e^{-\beta  E^{(o)}(\vec x)},
    \end{equation}
    where $\mathcal N_\beta^{-1} = \sum_{\vec x \in \{0,1\}^n} e^{-\beta E(\vec x)}$ is the normalization constant of the pmf. We then define the minimal objective energy as $E_\mathrm{min} = \min_{\vec x \in \{0,1\}^n} E^{(o)}(\vec x)$ and bound
    \begin{equation}
        \mathbb P[\vec x \notin \mathcal F] \le \mathcal N_\beta e^{-\beta (M+ E_\mathrm{min})} \sum_{\vec x \notin \mathcal F} 1 \le \mathcal N_\beta e^{-\beta (M + E_\mathrm{min})} 2^n.
    \end{equation}
    To bound the normalization constant, suppose we know the energy of a feasible bitstring to be $E(\vec x_\mathrm{feas}) = E_\mathrm{feas}$, which doesn't depend on $M$ since $E(\vec x) = E^{(o)}(\vec x)$ if $\vec x \in \mathcal F$; then, $ \mathcal N _\beta^{-1} = \sum_{\vec x \in \{0,1\}^n} e^{-\beta E(\vec x)} \ge e^{-\beta E_\mathrm{feas}}$ and
    \begin{equation}
        \mathbb P[\vec x \notin \mathcal F] \le e^{-\beta (M + E_\mathrm{min} - E_\mathrm{feas})} 2^n
    \end{equation}
    To ensure that $\mathbb P[\vec x \notin \mathcal F] \le 1-\eta$ it is thus sufficient to pick $M$ such that $e^{-\beta (M + E_\mathrm{min} - E_\mathrm{feas})} 2^n \le 1-\eta$, or, equivalently,
    \begin{equation} \label{eq:thermal_M_our}
         M \ge \beta^{-1} \bigl(n\ln(2) - \ln(1-\eta)\bigr) + E_\mathrm{feas} - E_\mathrm{min}
    \end{equation}
    To prove the claim on $M_{\ell_1}(\beta)$, it is left to show only that $ \norm{Q}_{\ell_1} \ge E_\mathrm{feas} - E_\mathrm{min}$. This is immediate as
    \begin{equation}
        \norm{Q}_{\ell_1} = \sum_{i,j} |Q_{ij}| = \sum_{i,j: Q_{ij} \ge 0} Q_{ij} - \sum_{i,j: Q_{ij} < 0} Q_{ij}
    \end{equation}
    Clearly, $\sum_{i,j: Q_{ij} \ge 0} Q_{ij} \ge \max_{\vec x} \vec x^t Q \vec x \ge E_\mathrm{feas}$; on the other hand, $\sum_{i,j: Q_{ij} < 0} Q_{ij} \le \min_{\vec x} \vec x^t Q \vec x = E_\mathrm{min}$. This shows that $\norm{Q}_{\ell_1} \ge E_\mathrm{feas}- E_\mathrm{min}$ and proves the lemma.
\end{proof}%
Notice that the fact that $ \norm{Q}_{\ell_1} \ge E_\mathrm{feas} - E_\mathrm{min}$ is related to the advantage for exact solvers in using the more conservative weight $M_\mathrm{SDP} = E_\mathrm{feas} - E_\mathrm{LB}$, rather than $M_{\ell_1} = M_{\ell_1}(\beta = \infty) = \norm{Q}_{\ell_1}$ as shown in Ref.~\cite{BigM_25}, where $E_\mathrm{LB}$ is a lower bound of $E_\mathrm{min}$.

\section{Triviality and existence of the solution} \label{app:eta_exist}
Sampling feasible solutions with a maximal allowed energy, with probability at least $\eta$ and using a Gibbs solver at temperature $\beta$, can range from extremely difficult to trivial. In extreme cases, the task may even become  infeasible. In this section, we illustrate how the algorithm handles these two opposite cases.

For a fixed $\beta$, there is an upper bound on the sampling probability of feasible solutions with energy not exceeding $E_f$. 
This bound corresponds to the probability of sampling such solutions assuming that only feasible points could be drawn, that is, in the limit of very large $M$ where unfeasible configurations are fully suppressed. 
Formally, this upper limit is the cumulative distribution function of the Gibbs measure restricted to the feasible set, evaluated at $E_f$.
Recall that in \cref{alg:M} we determine $M$ from three bounds using 
\begin{equation} \label{eq:master_function}
    g(M) = \Binf(M) + \Bfh - \frac{1-\eta}{\eta}\Bfl
\end{equation}
The function $g$ is monotonically decreasing in $M$, and it is constructed so that its zeros (and the region where $g(M) \le 0$) correspond to values of $M$ satisfying the requirement. The maximal $\eta$ ensuring the existence of a guaranteed sampling success, which we denote as $\eta_\mathrm{exist}$, occurs when $g(M)$ tends to zero only asymptotically. Since $g(M) \stackrel{M \to \infty}{\longrightarrow} \Bfh - \frac{1-\eta}{\eta}\Bfl$, then by setting this limit to zero we obtain
\begin{equation} \label{eq:eta_ex_opt}
    \eta_\mathrm{exist}  =  \frac{\Bfl }{\Bfl + \Bfh}.
\end{equation}
In the algorithm, if the required sampling success satisfies $\eta \ge \eta_\mathrm{exist}$, then the requirement is unattainable by the algorithm: $g$ has no roots and the algorithm returns $\{\}$. The required sampling success can then be reduced to $\eta_\mathrm{exist} - \epsilon$ for a small $\epsilon$ and the algorithm rerun to obtain an attainable solution. 
Conversely, if the requirement is already met even with no penalty weight ($M=0$), the problem is considered trivial: the root of $g$ will be negative and therefore returning $\max \{0, M^\ast\}$ will ensure a positive and sufficient penalty weight.

\section{Robust implementation in logspace} \label{app:logspace_implementation}
A naive numerical implementation of the algorithm may suffer from numerical instabilities or overflow, primarily due to the exponentials involved in the computations of the terms $\Binf$, $\Bfl$ and $\Bfh$. 
To prevent this, we implement a numerically stable formulation,  described in this section.

Let us first take the function  \eqref{eq:master_function}, and recall that $\Binf(M) = \sum_{v=1}^{v_\mathrm{cut}} e^{-\beta M v} n_{\mathrm{pen}}(v)$, $\Bfh = \sum_{e \in \LambdaHigh} \e^{-\beta e} n_\Delta(e + \ELB)$, $\Bfl = \sum_{e \in \LambdaLow} \e^{-\beta (e+\Delta)} n_\Delta(e + \ELB)$ and set $c = \frac{1 - \eta}{\eta}$. 
We can alternatively define the function
\begin{equation}
    G(M) = \log(\Binf(M) + \Bfh) - \log\left(\frac{1-\eta}{\eta}\right) - \log(\Bfl).
\end{equation}
Note that the functions $g(M)$ and $G(M)$ share the same sign and root. 
hence, $G(M)$ can be used in place of $g(M)$ in the \cref{alg:M} yielding logarithmically smaller values in the evaluation. 
To compute the terms of $G(M)$ efficiently and stably, we introduce the LogSumExp (LSE) function \cite{Boyd2009convexoptimization}, defined as $\mathrm{LSE}(x_1,\dots,x_n) = \log\big(\sum_{i=1}^n e^{x_i}\big)$, which is commonly used to avoid underflow and overflow, among other things. 
Using the LSE, we compute the terms of $G$ as
\begin{align}
    \log(\Binf(M)) &= \mathrm{LSE}[\log(n_{\mathrm{pen}}(v)) - \beta Mv]_{v=1}^{v_\mathrm{cut}}\,, \\
    \log(\Bfh) &= \mathrm{LSE}[\log(n_\Delta(e + \ELB)) - \beta e]_{e \in \LambdaHigh}\,, \\
    \log(\Binf(M)  + \Bfh) &= \mathrm{LSE}[\log(\Binf(M)), \log(\Bfh)]\,, \\
    \log(\Bfl) &= \mathrm{LSE}[\log(n_\Delta(e + \ELB)) - \beta (e+\Delta)]_{e \in \LambdaLow}\,.
\end{align}

Using the log-domain function $G(M)$ instead of $g(M)$ also affects the computation of the existence threshold discussed in \cref{app:eta_exist}. 
By imposing $G(\infty) = 0$, one can compute $\eta_\mathrm{exist} = (1 + e^{\gamma})^{-1}$,
with $\gamma = \log(\Bfh) - \log(\Bfl)$, properly evaluated using the LSE function to ensure numerical stability.

\section{Dependence of the Algorithm's output on the input parameters (\texorpdfstring{$v_\mathrm{cut}$}{vcut} and \texorpdfstring{$E_f$}{Ef})} \label{app:output_dependence}
\begin{figure*}[tb]
    \centering
        \includegraphics[width=.8\linewidth]{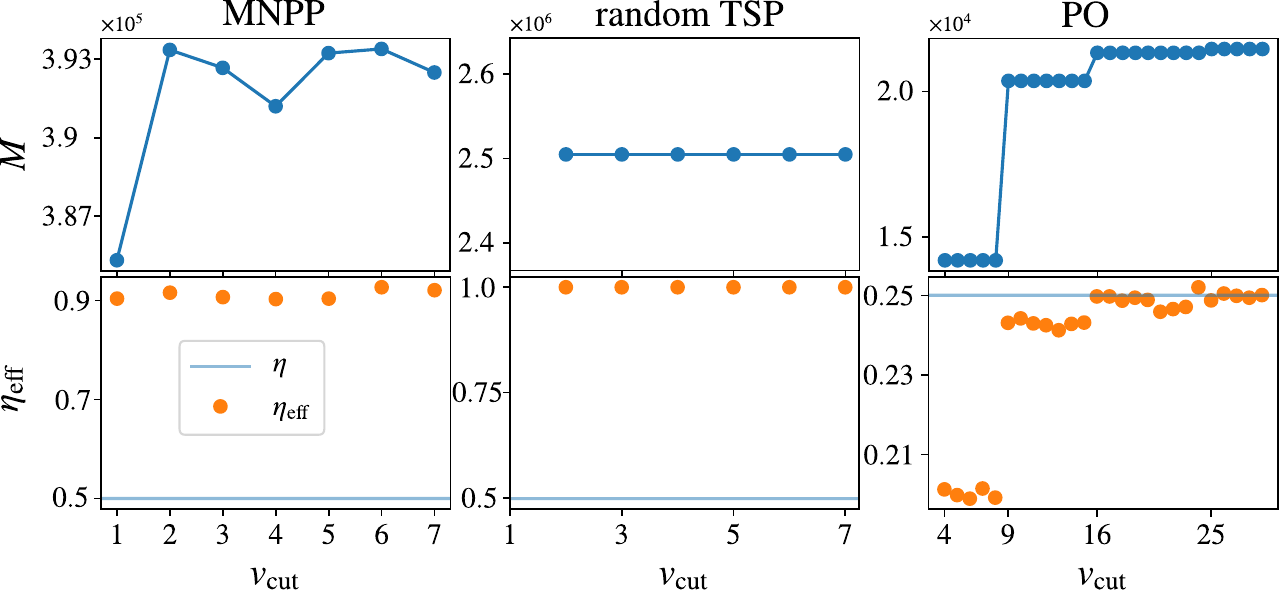}
    \caption{\emph{Output dependence on $v_\mathrm{cut}$.}     
     Algorithm output $M^*$ (top) and associated effective success probability $\eta_\mathrm{eff}$ (bottom) as functions of the hyperparameter $v_\mathrm{cut}$ for the tested benchmarks, shown for small system sizes. In the bottom plots, the thin blue line denotes the required probability $\eta$ and orange markers the effective success $\eta_\mathrm{eff}$ measured from an ideal Gibbs sampler. For \ac{MNPP} and \ac{TSP}, $M^*$ stabilizes already at small $v_\mathrm{cut}$, with $\eta_\mathrm{eff} > \eta$ in the corresponding plots—indicating that small $v_\mathrm{cut}$ values capture most infeasible samples and ensure a robust guarantee that is respected. In contrast, for \ac{PO} instances, values $v_\mathrm{cut} \ge 16$ are required to capture all relevant infeasible configurations, stabilize $M^*$, and and achieve an $\eta_\mathrm{eff}$ close to the target $\eta$. 
     The required probability $\eta$ has been set to $0.5$ ($0.25$ for \ac{PO}), the inverse temperature to $\beta = 10^{-5}$ and the problem-specific parameters to $N,P = 4,3$ (\ac{MNPP}), $n_c = 4$ (\ac{TSP}) and $N,w = 4,3$ (\ac{PO}).
     }
    \label{fig:vcut_dependence}
\end{figure*}

\begin{figure*}[tb]
    \centering
        \includegraphics[width=.6\columnwidth]{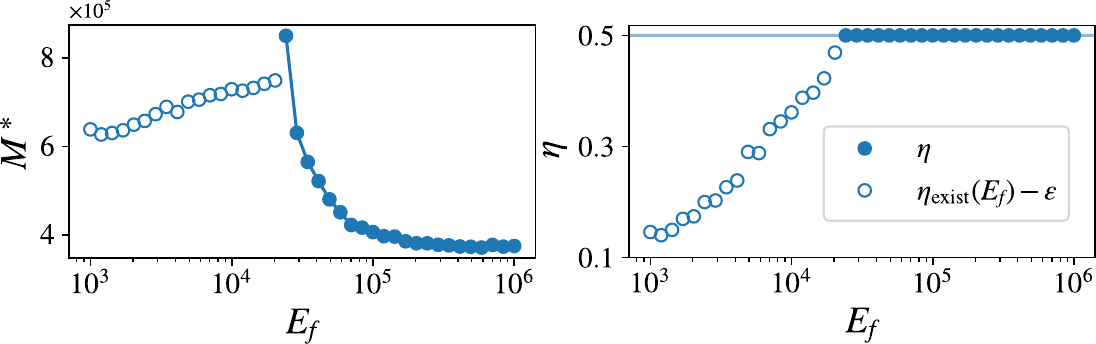}
    \caption{\emph{Output dependence on $E_f$.}
    Algorithm output $M^*$ (left) and success probability used for the output computation $\eta$ (right) as functions of the input parameter $E_f$, for a small instance of \ac{MNPP} ($N,P = 9,2$), illustrating the parameter's effect. Filled blue markers denote the region where the required probability $\eta$ can be satisfied and thus the output $M^\ast$ is computed using the required $\eta$, while empty markers denote the region where the requirement becomes unattainable, and the guarantee is thus reduced to $\eta = \eta_\mathrm{exist}(E_f) - \epsilon$, where $\eta_\mathrm{exist}$ is the maximal probability for which the existence of a solution $M$ is ensured
    (see \cref{app:eta_exist}). The required probability is set to $\eta =0.5$, the inverse temperature to $\beta = 10^{-5}$ and $\epsilon = 0.01$.
    }
    \label{fig:Ef_dependence}
\end{figure*}
This section provides a concise analysis of the sensitivity of the algorithm’s output to choices of the input parameters $E_f$, the maximal energy of desired solutions, and $v_\mathrm{cut}$, the maximal value of the penalization term considered for infeasible points. 
\cref{fig:vcut_dependence} illustrates, for small instances ($n_\mathrm{bits} \approx 16$) across all benchmarked problems, the behavior of the returned weight $M^*$ and its associated required success probability $\eta$, compared with the effective success probability $\eta_\mathrm{eff}$ of an ideal Gibbs sampler. 
For the \ac{MNPP} and \ac{TSP} instances, small values of $v_\mathrm{cut}$ do not change  $M^*$ accounting for most of the probability mass of infeasible points, as indicated by $\eta_\mathrm{eff} \ge \eta$, confirming the robustness of the guarantee. 
In contrast, for \ac{PO}, values $v_\mathrm{cut} < 16$ fail to capture all relevant infeasible configurations. 
Here $M^*$ increases stepwise (because $n_\mathrm{pen}(v)$ is non-zero only for perfect squares $v$) and $\eta_\mathrm{eff} < \eta$, indicating an unfulfilled guarantee. 
Only for $v_\mathrm{cut} \ge 16$, $M^*$ converges and the effective success probablility is close to its target. 
Similar behaviour was observed for larger problem sizes. Accordingly, $v_\mathrm{cut}$ was fixed to $4$ for \ac{MNPP} and \ac{TSP}, and to $16$ for \ac{PO} instances.

\cref{fig:Ef_dependence} shows the dependence of $M^*$ and $\eta$ on the parameter $E_f$ for an exemplary \ac{MNPP} instance with $18$ bits. 
As $E_f$ decreases, $M^*$ increases to ensure a success probability of $\eta = 0.5$, as long as this target remains  achievable (filled markers). 
For sufficiently small $E_f$, however, the requirement becomes unattainable. 
In this regime, the algorithm is re-run with $\eta = \eta_\mathrm{exist} - \epsilon$, where $\eta_\mathrm{exist}$ is the maximal success probability for which the existence of a solution $M$ is ensured; see \cref{app:eta_exist} for details.
When this happens (empty markers), both the maximal attainable success $\eta_\mathrm{exist}$ and $\eta$ alongside, decrease as $E_f$ is lowered, as expected.

\section{Solving the inverse problem as a byproduct: from \texorpdfstring{$M$}{M} to \texorpdfstring{$\beta$}{β}} \label{app:inverse_problem}
A convenient extension of the algorithms introduced in this work arises in the inverse problem to the Big-$M$ problem: determine the inverse temperature $\beta$ that allows a thermal solver to sample desired solutions with probability at least $\eta$, when the penalty weight $M$ is fixed. 
Such a situation may occur when increasing $M$ is not possible due to interaction strengths, hardware, or software limitations, while the solver temperature can still be adjusted to meet the sampling requirement. 
This inverse problem can be solved with the same strategy introduced in this work (Sec.~\ref{sec:recipe}). 
To this end, we modify \cref{alg:M} line~\ref{alg:line:g_root} to calculate an optimal inverse temperature $\beta^\ast$ as the root of $\beta \mapsto g(M)$ for given $M$.

\section{Penalization degeneracy} \label{app:penalization_degeneracies}
In this section, we present the analytical results on the penalization degeneracy $n_\mathrm{pen}(v) = |\{\vec x \in \{0,1\}^n : E^{(p)}(\vec x) = v\}|$ for the constrained optimization problems considered in this work. 
For the benchmarked problems \ac{MNPP}, \ac{TSP} and \ac{PO}, the penalization degeneracy can be obtained analytically through combinatorial counting methods. The following subsections detail these analytical expressions 
for each problem. 
Note that alternatively one can also estimate $n_\mathrm{pen}(v)$ for $v \in \{1,\dots,v_\mathrm{cut}\}$ by uniformly sampling bitstrings $\vec x \in \{0,1\}^n$, evaluating their penalization energy $E^{(p)}(\vec x)$ and then fitting or inferring the corresponding values.

Importantly, the analytical argument used to derive the penalization degeneracies $n_\mathrm{pen}(v)$, suggests that, for any fixed value of $v$, computing $n_{\mathrm{pen}}(v)$ takes at most polynomial time and memory in $n$. 
For the benchmark problems, \ac{MNPP}, \ac{TSP} and \ac{PO}, the entries of the constraint matrix $A$ are upper bounded by $\mathcal O(1)$. Consequently, as stated in the main text, the penalization term is bounded 
$v_\mathrm{max} = \max_\vec x E^{(p)}(\vec x) \in \mathcal O (mn^2)$. 
Therefore, evaluating $n_{\mathrm{pen}}(v)$ for all $v$ also takes overall polynomial time and memory in the problem size. 
Expressed in terms of the parameters of each problem, $v_\mathrm{max}$ reads as follows: for \ac{TSP}, we have $m = n_v$, the number of vertices, and $n = n_v^2$, hence $v_\mathrm{max} \in \mathcal O (n_v^5)$; for \ac{MNPP} ($N$ numbers in $P$ partitions), we have $m = P$ and $n = NP$, hence $v_\mathrm{max} \in \mathcal O (N^2P^3)$; 
for \ac{PO}, we have $m = 1$ and $n = wN$, the resolution $w$ times number of stocks $N$,  hence $v_\mathrm{max} \in \mathcal O (w^2N^2)$.

In \cref{fig:degeneracies} we illustrate the results of the remainder of this appendix.  
We show the penalization degeneracies computed analytically for small values of the penalization at a fixed system size. 
We observe that $n_\text{\rm pen}$ scales sub-exponentially for all three problems.  
This indicates that one can often introduce rather small 
values of the cut-off $v_\text{\rm cut}$ until which $n_\text{\rm pen}$ is computed in our \cref{alg:M}.

\subsection{MNPP}
For a Multiway Number Partition Problem (MNPP) instance with $N$ numbers to partition into $P$ partitions the penalization $E^{(p)}(\vec x) = \sum_{i=1}^{N}  ( 1- \sum_{p=1}^{P} x_{i, p} )^2$ enforces a feasible decision matrix $[\vec x]_{i,p}$ to be right-stochastic, i.e., $\sum_{p=1}^P x_{i,p} = 1 \; \forall i$, meaning that each row of the binary matrix contains exactly one $1$. 
Note that in the penalization term the rows contribute independently to the total error. 
Therefore, the number of feasible bitstrings ($v=0$) is $n_\mathrm{pen}(0)=P^N$, since there are $P$ valid choices per row and $N$ rows. 
Similarly, the number of bitstrings with $v=1$ can be counted by considering the number of configurations with exactly one incorrect row and multiplying by the number of ways this row can contain either $0$ or $2$ ones instead of $1$, thus, giving $E^{(p)}(\vec x) =1$. 
The same reasoning extends to higher values ($v = 2,3$).  
For $v = 4,\dots,7$ one must also include rows with $3$ ones instead of $1$, which contribute a factor of four to $E^{(p)}(\vec x)$. 
In this way, the following expressions are obtained:

\begin{align}
    n_\mathrm{pen}(0) &= P^N \:\:\:\text{ feasible bitstrings}\\
    n_\mathrm{pen}(v) &= P^{N-v} \binom{N}{v} \biggl(1 + \binom{P}{2}\biggr)^v \;\; \forall v \in \{1,\cdots, 3\} \\
    n_\mathrm{pen}(v) &= P^{N-v} \binom{N}{v} \biggl(1 + \binom{P}{2}\biggr)^v + (v-3)\binom{N}{v-3}P^{N-(v-3)}\binom{P}{3}\biggl(1 + \binom{P}{2}\biggr)^{v-4}  \;\; \forall v \in \{4,\cdots, 7\}
\end{align}

\subsection{TSP}
For a Travelling Salesman Problem instance with $n_v$ vertices, the penalization $E^{(p)}(\vec x) = \sum_{i=1}^{n_v}  ( 1- \sum_{t=1}^{n_v} x_{t, i} )^2 + \sum_{t=1}^{n_v}  ( 1- \sum_{i=1}^{n_v} x_{t, i} )^2$ enforces a feasible decision matrix $[\vec x]_{t,i}$ to be stochastic, meaning that each row \emph{and} each column must contain exactly a single $1$. The number of feasible points thus corresponds to the number of permutation matrices, hence $n_\mathrm{pen}(0) = n_v( n_v-1)\dots 1 =  n_v!$. Unlike the \ac{MNPP} case, here row and column errors are coupled (e.g., flipping a bit in a valid configuration introduces an error in both a row and a column), so their contributions to the total penalty are not independent. 
Consequently, no configurations exist with $v=1$. 
Indeed, having a total number of ones different from $n_v$ inevitably produces at least one faulty row (with zero or two ones) \emph{and} one faulty column, producing a minimum of $v=2$. 
Conversely, if the matrix contains exactly $n_v$ ones, faulty rows (and columns) must occur in pairs, so $v$ is always even. 
Similarly, all odd-valued violations considered (up to $v=7$) yield no configuration. 
The coupling between rows and columns rapidly complicates the combinatorial counting as $v$ increases. 
To determine $n_\mathrm{pen}(v)$, one must first identifies all distinct patterns in which a given $v$ can occur, and then count the configurations realizing each pattern. 
For instance, for $v=2$ we must account for several distinct scenarios: matrices with two faulty rows and no faulty columns, yielding 
$n_v! \binom{n_v}{2}$ configurations; matrices with two faulty columns and no faulty rows (the same number by symmetry); matrices with one faulty row and one faulty column containing $n_v-1$ ones in total (rather than $n_v$ like a valid configuration), that are 
$n_v! n_v$; 
matrices with one faulty row and one faulty column with $n_v+1$ ones in total and a $1$ at the intersection of the faults, giving 
$n_v! 2\binom{n_v}{2}$; and finally, matrices with one faulty row and one faulty column containing $n_v+1$ ones in total but a $0$ at the intersection of the faults, resulting in 
$n_v! \frac{3}{2}\binom{n_v}{3}$. 
Analogous decompositions, based on the number of ones and the distribution of horizontal or vertical faults, were performed to avoid overcounting and to compute degeneracies up to $v = 6$:

\begin{align}
    n_\mathrm{pen}(0) &= n_v! \:\:\:\text{ feasible bitstrings}\\
    n_\mathrm{pen}(1) &= 0 \\
    n_\mathrm{pen}(2) &= n_v!  \biggl[n_v + 4\binom{n_v}{2} + \frac{3}{2}\binom{n_v}{3} \biggr]  \\
    n_\mathrm{pen}(3) &= 0 \\
    n_\mathrm{pen}(4) &= n_v! \biggl[\binom{n_v}{2}+21\binom{n_v}{3}+57\binom{n_v}{4}+45\binom{n_v}{5}+\frac{45}{4}\binom{n_v}{6}\biggr] \\
    n_\mathrm{pen}(5) &= 0 \\
    n_\mathrm{pen}(6) &= 5n_v!\biggl[\frac{47}{15}\binom{n_v}{3}+ 24\binom{n_v}{4}+ 137\binom{n_v}{5}+ 1157\binom{n_v}{6}+ \frac{567}{4}\binom{n_v}{7}+ 126\binom{n_v}{8}+ \frac{63}{2}\binom{n_v}{9}\biggr] \\
    n_\mathrm{pen}(7) &= 0
\end{align}

\subsection{PO}
For a Portfolio Optimization instance with $N$ stocks and partition number $w$, the penalization energy is $E^{(p)}(\vec x) = ( \sum_{i=1}^N x_i - 2^w + 1 )^2$, where each component is an integer $x_i \in \{0,\dots,2^w-1\}$. The penalization is zero only for bitstrings $\vec x$ whose components sum up to $2^w-1$. Since $E^{(p)}$ is the square of an integer, non-zero configurations exist only for perfect-square values $v=k^2$. 
The number of feasible configurations ($v=0$) can be computed as the number of ways to distribute $2^w-1$ indistinguishable `coins' (portfolio fragments) among $N$ distinguishable `boxes' (stocks) by standard combinatorial arguments. 
For infeasible bitstrings, a similar reasoning applies: if the configuration components sum up to $2^w-1 \pm k$, then $v=k^2$. 
The corresponding number of configurations is the sum of two terms: the ways to distribute $2^w-1 -k$ (indistinguishable) coins among $N$ (distinguishable) boxes, and the ways to distribute $2^w-1+k$ coins among $N$ boxes, constrained so that no box can contain more than $2^w-1$ coins, since each $x_i$ is bounded by that value. 
This procedure yields the following penalization degeneracies:

\begin{align}
    n_\mathrm{pen}(0) &= \binom{2^w + N - 2}{N-1} \:\:\:\text{ feasible bitstrings}\\
    n_\mathrm{pen}(1) &= \binom{2^w + N - 3}{N-1} +\binom{2^w + N - 1}{N-1} - N \\
    n_\mathrm{pen}(4) &= \binom{2^w + N - 4}{N-1} +\binom{2^w + N}{N-1} - N^2\\
    n_\mathrm{pen}(9) &= \binom{2^w + N - 5}{N-1} +\binom{2^w + N + 1}{N-1} - N^2\frac{N+1}{2} \; \textit{ for } w \ge 2\\
    n_\mathrm{pen}(16) &= \binom{2^w + N - 6}{N-1} +\binom{2^w + N + 2}{N-1} - N^2\frac{N+1}{2}\frac{N+2}{3} \; \textit{ for } w \ge 2\\
    n_\mathrm{pen}(25) &= \binom{2^w + N - 7}{N-1} +\binom{2^w + N + 3}{N-1} - N^2\frac{N+1}{2}\frac{N+2}{3}\frac{N+3}{4} \; \textit{ for } w \ge 3\\
    n_{pen}(v) &= 0 \; \forall v :  \nexists k \in \mathbb{N} : v = k^2 
\end{align}

\begin{figure*}
    \centering
        \includegraphics[width=\linewidth]{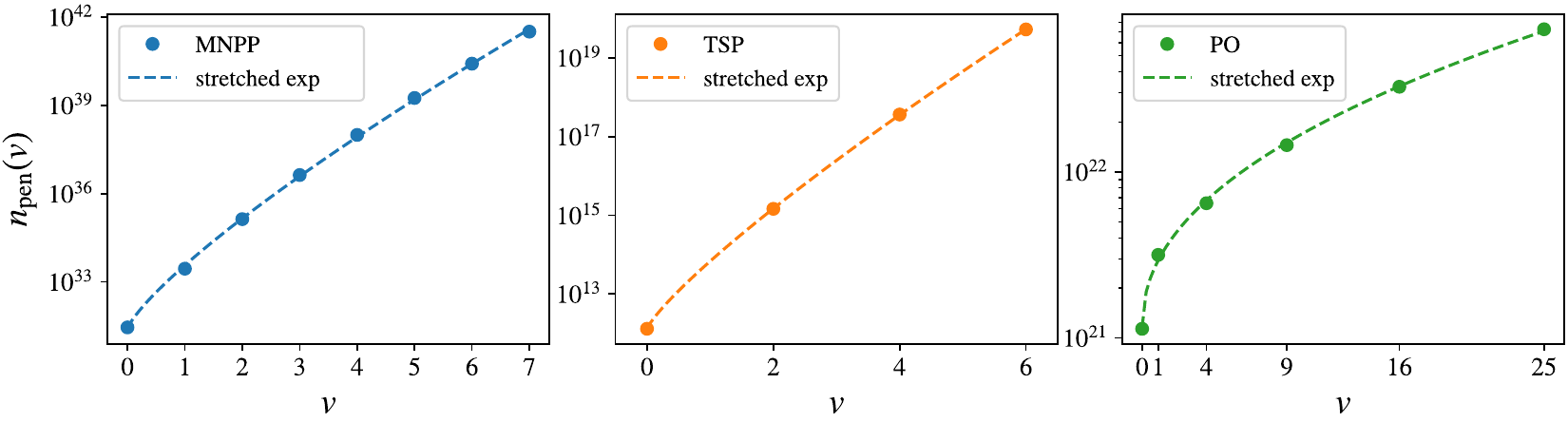}
    \caption{\emph{Penalization degeneracies for the problems studied and stretched exponential fit.}  $n_\mathrm{pen}(v)$ is shown as a function of the violation $v$ on a logarithmic scale for the benchmarked problems (from left to right, \ac{MNPP}, \ac{TSP} and \ac{PO}) at a fixed problem size. The problem parameters were chosen such that each instance requires $255$ bits; specifically, for \ac{MNPP} $N=45$, $P=5$, for \ac{TSP} $n_v = 25$ and for \ac{PO} $N=45$, $w=5$. Dots represent the analytical values reported in \cref{app:penalization_degeneracies}, while dashed lines show a stretched exponential fit, defined as $n_\mathrm{pen}(v) = \exp(a + bx^k)$. Note that for $k =1$ the stretched exponential reduces to a standard exponential, whereas $k < 1$ corresponds to sub-exponential growth. The fit closely follows the data, with fitted exponents for the three cases (from left to right) given by $k_\mathrm{MNPP} = 0.79 \pm 0.05$, $k_\mathrm{TSP} = 0.83 \pm 0.02$ and $k_\mathrm{PO} = 0.46 \pm 0.05$, where the errors correspond to three standard deviations. In all cases, $n_\mathrm{pen}(v)$  exhibits sub-exponential scaling with $v$. 
    }
    \label{fig:degeneracies}
\end{figure*}

\end{document}

%% file: myacronyms.tex
\begin{acronym}[POVM]
\acro{ACES}{averaged circuit eigenvalue sampling}
\acro{AGF}{average gate fidelity}
\acro{AP}{Arbeitspaket}
\acro{ASS}{algorithmic shadow spectroscopy}

\acro{BK}{Bravyi-Kitaev}
\acro{BOG}{binned outcome generation}

\acro{CFG}{context-free grammar}
\acro{CFL}{context-free language}
\acro{CNF}{conjunctive normal form}
\acro{CP}{completely positive}
\acro{CPT}{completely positive and trace preserving}
\acro{cs}{computer science}
\acro{CS}{compressed sensing} 
\acro{ctrl-VQE}{ctrl-VQE}

\acro{DA}{Digital Annealer}
\acro{DAU}{Digital Annealer Unit}
\acro{DAQC}{digital-analog quantum computing}
\acro{DD}{dynamical decoupling}
\acro{DFA}{deterministic finite automaton}
\acrodefplural{DFA}{deterministic finite automata}
\acro{DFE}{direct fidelity estimation} 
\acro{DFT}{discrete Fourier transform}
\acro{DL}{deep learning}
\acro{DM}{dark matter}

\acro{FFT}{fast Fourier transform}

\acro{GS}{ground-state}
\acro{GST}{gate set tomography}
\acro{GTM}{gate-independent, time-stationary, Markovian}
\acro{GUE}{Gaussian unitary ensemble}

\acro{HOG}{heavy outcome generation}

\acro{irrep}{irreducible representation}

\acro{LBA}{linear bounded automaton}
\acrodefplural{LBA}{linear bounded automata}
\acro{LCBO}{linearly constrained binary optimization}
\acro{LDPC}{low density partity check}
\acro{LP}{linear program}

\acro{MAGIC}{magnetic gradient induced coupling}
\acro{MAX-SAT}{maximum satisfiability}
\acro{MBL}{many-body localization}
\acro{MIP}{mixed integer program}
\acro{ML}{machine learning}
\acro{MLE}{maximum likelihood estimation}
\acro{MMP}{matrix mortality problem}
\acro{MNPP}{Multiway Number Partitioning Problem}
\acro{MPO}{matrix product operator}
\acro{MPS}{matrix product state}
\acro{MS}{M{\o}lmer-S{\o}rensen}
\acro{MSE}{mean squared error}
\acro{MTM}{Multi-tape Turing machine}
\acro{MUBs}{mutually unbiased bases} 
\acro{mw}{micro wave}

\acro{NFA}{non-deterministic finite automaton}
\acrodefplural{NFA}{non-deterministic finite automata}
\acro{NISQ}{noisy and intermediate scale quantum}
\acro{NTM}{non-deterministic Turing machine}
\acro{PTM}{probabilistic Turing machine}

\acro{ONB}{orthonormal basis}
\acroplural{ONB}[ONBs]{orthonormal bases}

\acro{PCP}{Post correspondence problem}
\acro{PDA}{pushdown automaton}
\acro{PO}{Portfolio Optimization}
\acro{POVM}{positive operator valued measure}
\acro{PQC}{parametrized quantum circuit}
\acro{PSD}{positive-semidefinite}
\acro{PSR}{parameter shift rule}
\acro{PVM}{projector-valued measure}

\acro{QAOA}{Quantum Approximate Optimization Algorithm}
\acro{QBF}{quantified Boolean formula}
\acro{QC}{quantum computation}
\acro{QEC}{quantum error correction}
\acro{QFT}{quantum Fourier transform}
\acro{QM}{quantum mechanics}
\acro{QML}{quantum machine learning}
\acro{QMT}{measurement tomography}
\acro{QPT}{quantum process tomography}
\acro{QPU}{quantum processing unit}
\acro{QUBO}{Quadratic Unconstrained Binary Optimization}
\acro{QWC}{qubit-wise commutativity}
\acro{LCBO}{linearly constrained binary optimization}
\acro{RB}{randomized benchmarking}
\acro{RBM}{restricted Boltzmann machine}
\acro{RDM}{reduced density matrix}
\acro{rf}{radio frequency}
\acro{RIC}{restricted isometry constant}
\acro{RIP}{restricted isometry property}
\acro{RMSE}{root mean squared error}
\acro{ROBP}{read once branching program}

\acro{SA}{simulated annealing}
\acro{SDP}{semidefinite program}
\acro{SFE}{shadow fidelity estimation}
\acro{SIC}{symmetric, informationally complete}
\acro{SPAM}{state preparation and measurement}
\acro{SPSA}{simultaneous perturbation stochastic approximation}

\acro{TM}{Turing machine}
\acro{TSP}{Traveling Salesman Problem}
\acro{TT}{tensor train}
\acro{TV}{total variation}

\acro{VC}{vertex cover problem}
\acro{VQA}{variational quantum algorithm}
\acro{VQE}{variational quantum eigensolver}

\acro{XEB}{cross-entropy benchmarking}

\end{acronym}